\documentclass[DIV=15,11pt,abstract]{scrartcl}

\title{Duality of Orthogonal and Symplectic Random Tensor Models}
\date{}
\usepackage{authblk}
\author[1,2,3]{Razvan Gurau}
\author[1]{Hannes Keppler}
\affil[1]{\normalsize\itshape 
	Heidelberg University, Institut f{\"u}r Theoretische Physik, Philosophenweg 19, 69120 Heidelberg, Germany
	\authorcr \hfill}
\affil[2]{\normalsize \itshape 
	CPHT, CNRS, Ecole Polytechnique, Institut Polytechnique de Paris, Route de Saclay, \authorcr 91128 PALAISEAU, 
	France
	\authorcr \hfill}
\affil[3]{\normalsize\itshape 
	Perimeter Institute for Theoretical Physics, 31 Caroline St. N, N2L 2Y5, Waterloo, ON,
	Canada
	\authorcr \hfill
	\authorcr
	Emails: gurau@thphys.uni-heidelberg.de, keppler@thphys.uni-heidelberg.de
	\authorcr \hfill}

\usepackage[T1]{fontenc}
\usepackage[utf8]{inputenc}
\usepackage[english]{babel}

\usepackage{lmodern}
\usepackage{microtype}

\setkomafont{title}{\large\rmfamily\bfseries}
\setkomafont{section}{\normalfont\rmfamily\bfseries\large}
\setkomafont{subsection}{\normalfont\rmfamily\bfseries}
\setkomafont{subsubsection}{\normalfont\rmfamily\bfseries\normalsize}
\setkomafont{paragraph}{\normalfont\rmfamily\bfseries}
\setkomafont{subparagraph}{\normalfont\itshape}

\RedeclareSectionCommands[    
tocentryformat=\normalfont\bfseries,
tocpagenumberformat=\normalfont\bfseries
]{section}
\RedeclareSectionCommands[    
tocentryformat=\normalfont,
tocpagenumberformat=\normalfont
]{subsection,subsubsection,paragraph,subparagraph}

\usepackage[intlimits,sumlimits]{amsmath}
\usepackage{amssymb, commath, mathtools, bm, amsthm, thmtools, relsize}
\usepackage[mathscr]{eucal}    
\usepackage{dsfont}
\usepackage{tensor}

\numberwithin{equation}{section}

\usepackage{booktabs} 
\usepackage{tabu}

\usepackage{tikz}
\usetikzlibrary{tikzmark, cd, positioning}

\usepackage{graphicx} 
\usepackage{blindtext} 
\usepackage[format=plain, labelfont=bf, labelsep=period]{caption} 

\usepackage[bottom]{footmisc}   

\usepackage[toc,page]{appendix}

\usepackage[
style=nature,         
backend=bibtex,
defernumbers=true,    
sorting=none,		  
isbn=false,
date=year,
url=false,
doi=false,
hyperref = true
]{biblatex}
\DeclareFieldFormat[article]{volume}{\textbf{#1}}    
\usepackage{xcolor}
\definecolor{BScRed}{RGB}{157,0,0}
\definecolor{BScBlue}{RGB}{1,1,141}
\definecolor{BScGreen}{RGB}{0, 119, 85}

\usepackage[colorlinks=True]{hyperref}
\hypersetup{allcolors=BScBlue}

\newtheoremstyle{plain}
{}{}
{\itshape}{}{\bfseries}
{\ifx\thmnote\@gobble\else\normalfont\fi.}
{.5em}{}

\newtheoremstyle{definition}
{}{}
{\normalfont}{}{\bfseries}
{\ifx\thmnote\@gobble\else\normalfont\fi.}
{.5em}{}

\newtheoremstyle{remark}
{}{}{\normalfont}{}{\itshape}{.}{.5em}{}

\theoremstyle{plain}    
\newtheorem{theorem}{Theorem}
\newtheorem{lem}{Lemma}

\newtheorem{prop}{Proposition}
\newtheorem{definition}{Definition}
\newtheorem{remark}{Remark}

\renewcommand{\del}{\partial}

\newcommand{\tr}{\mathrm{Tr}}

\newcommand{\isom}{\cong}
\DeclareMathOperator{\sgn}{sgn}
\newcommand{\suprho}{\raisebox{.5em}{{$\kern-.08em\mathsmaller\varrho$}}}

\begin{document}

{\hypersetup{allcolors=black}
\maketitle
\begin{abstract}
	The groups \(O(N)\) and \(Sp(N)\) are related by an analytic continuation to negative values of \(N\), \linebreak[4]\mbox{\(O(-N)\simeq Sp(N)\)}. This duality has been studied for vector models, \(SO(N)\) and \(Sp(N)\) gauge theories, as well as some random matrix ensembles. We extend this duality to real random tensor models of arbitrary order \(D\) with no symmetry under permutation of the indices and with quartic interactions. The \(N\) to \(-N\) duality is shown to hold graph by graph to all orders in perturbation theory for the partition function, the free energy and the connected two point function.
\end{abstract}
\microtypesetup{protrusion=false}
\setcounter{tocdepth}{2}
\tableofcontents
\microtypesetup{protrusion=true}
}


\section{Introduction and Conclusion}\label{sec: introduction}

Dualities are non trivial relations between seemingly different models and therefore of great use in physics and mathematics. It has been known for some time \cite{Mkrtchian} that, for even \(N\), \(SO(N)\) and \(Sp(N)\) gauge theories are related by changing \(N\) to \(-N\) and that one can make sense of the relation \mbox{\(SO(-N)\simeq Sp(N)\)} for the representations of the respective groups \cite{Cvitanovic}. This duality has furthermore been shown to hold between orthogonal and symplectic matrix ensembles  \cite{goe-gse}\footnote{These correspond to the \(O(N)\otimes O(N)\) and \(Sp(N)\otimes Sp(N)\) matrix models of Sec.\,\ref{sec: matrix}.}.

The \(N\) to \(-N\) duality inspired in part the conjectured holographic duality between Vasiliev's higher spin gravity \cite{Vasiliev} in four-dimensional de Sitter space and the three-dimensional euclidean \(Sp(N)\) vector model with anti commuting scalars \cite{dscft}. This  dS/CFT correspondence is in turn based on the conjectured Giombi-Klebanov-Polyakov-Yin duality \cite{KP,GY} relating the three dimensional \(O(N)\) vector model in the large \(N\) limit to Vasiliev gravity in four-dimensional anti-de Sitter space. In this context \(N\sim {(\Lambda G_N)}^{-1}\) so that the sign change of the cosmological constant \(\Lambda\) (holding \(G_N\) fixed) is accompanied by a change \(N\to -N\).

The perturbative expansion of random matrix models is a sum over \emph{ribbon graphs} representing topological surfaces. The weight of each graph is fixed by the Feynman rules and the perturbative series can be organized \cite{tHooft} as a topological expansion in \(1/N\). Random matrices yield a theory of random two-dimensional topological surfaces relevant for the study of conformal field theories (CFTs) coupled to two-dimensional Liouville gravity \cite{DiFrancesco,Kazakov,Douglas,Brezin,KPZ} and two-dimensional Jackiw-Teitelboim gravity \cite{Saad, Stanford, Johnson}. They have applications as combinatorial generating functions to several counting problems \cite{itzyksonzuber,zinn-justin,tanasabook} and to the intersection theory on the moduli space of Riemann surfaces \cite{penner,witten,kontsevich}.
 
Random matrices generalize to random tensor models \cite{Ambjorn,gurau,Gurau-inviation,Guraureview} of higher order\footnote{In the physics literature one often uses ``rank'' instead of order, but this may lead to confusion with the many notions of tensor rank in abstract algebra.} $D$ which are probability measures of the type:
\begin{equation*}
d\mu[T] = e^{-S[T]} \prod_{(a_1,\dots,a_D)}\frac{dT^{a_1\dots a_D}}{\sqrt{2\pi}} \; ,
\end{equation*}
where the action \(S[T]\) is build out of invariants under some symmetry transformation. These models can also be viewed as 0-dimensional quantum field theories. 
The Feynman graphs of such models can be interpreted as higher dimensional cellular complexes and the perturbative series can be reorganized as a series in $1/N$
\cite{Gurau-N,Bonzom:2012hw,Carrozza:2015adg,sabine,sylvan,Carrozza:2021qos} which is not topological for $D\ge 3$. 
Zero dimensional random tensors yield a framework for the study random topological spaces; in one dimension tensor models provide an alternative to the Sachdev-Ye-Kitaev model without quenched disorder \cite{Witten:2016iux}; in higher dimensions they lead to tensor field theories and a new class of large  \(N\) \emph{melonic} conformal field theories \cite{Giombi:2017dtl,Bulycheva:2017ilt,Giombi:2018qgp,Klebanov:2018fzb,Gurau-TFT}.

\paragraph{Main result.}
 In this paper we deal with tensors with $D$ indices
(i.\,e.\ of order $D$) with no symmetry under their permutations. The position of an index is called its color \(c\), with \(c=1,2,\dots D\). 
 The tensors transform in the tensor product of $D$ fundamental representations of \(O(N)\) and/or  \(Sp(N)\), i.\,e.\ each tensor index is transformed by a different \(O(N)\) or \(Sp(N)\) matrix.   
The tensor components are real gra\ss mann valued (anticommuting, odd) if the number of \(Sp(N)\) factors is odd and real bosonic (commuting, even) if this number is even\footnote{The tensors are even multilinear maps on \(\mathds{R}^{m|n}\), the real graded supervector space with \(m\) even and \(n\) odd directions. This is natural because the orthosymplectic super Lie group \(OSp(m,n)\) contains both \(O(m)\) and \(Sp(n)\) and acts on \(\mathds{R}^{m|n}\). This will be our guideline in constructing the models of interest.}. We assign a parity to the tensor indices: \(\abs{c}= 0\) or \(\abs{c} = 1\) if the index transforms under \(O(N_c)\) or \(Sp(N_c)\), respectively. We consider actions consisting in invariants up to quartic order (see Sec.\,\ref{sec: deftensor} for more details).

\begin{definition}\label{def: intro}
The real quartic graded tensor model, where ``graded'' refers to symmetry under:
\[
\pmb{O}_1(N_1)\otimes \pmb{O}_2(N_2)\otimes\dots\otimes \pmb{O}_D(N_D), \quad \pmb{O}_c(N_c) = \smash[tb]{\begin{cases}
O(N_c), \quad \abs{c}=0 \\
Sp(N_c), \quad \abs{c}=1
\end{cases}} \; ,
\]
is defined by the measure:
\begin{align*}
d\mu[T] \simeq e^{-S[T]}\ \prod_{a_1,\dots,a_D} dT^{a_1\dots a_D} \; ,
 \qquad
S[T]= \frac{1}{2} \Big(T^{a_1\dots a_D} T^{b_1\dots b_D} \prod_{c=1}^D g^c_{a_c b_c}\Big) + \sum_{q\in\mathscr{Q}} \frac{\lambda_q}{4} I^{q}(T) \; ,
\end{align*}
where \(g^c_{a_cb_c}\) is the Kronecker $\delta_{a_cb_c} $ for \(\abs{c}=0\) or the canonical symplectic form \(\omega_{a_cb_c}\) for \(\abs{c}=1\) and the sum over \(\mathscr{Q}\) runs over all the independent quartic trace invariants \(I^{q}(T)\).
\end{definition}

The partition function \(Z\) and the connected two-point function \(G_2\) of the model are defined by:
\[
Z(\lambda)=\int d\mu[T], \quad\text{and}\quad G_2( \lambda )=\frac{1}{Z} \int d\mu[T]\ T^{a_1\dots a_D} T^{b_1\dots b_D} \prod_{c=1}^D g^c_{a_c b_c} \; , \]
 and can be evaluated in a perturbative expansion.
Our main theorem is the following.
 
\begin{theorem}\label{thm: main}
The perturbative series of the free energy \(\ln Z\) and of the connected two point function \(G_2\) can be expressed as formal sums over connected, colored multi-ribbon graphs:
\begin{equation}\label{eq:combiintro}
\begin{split}
\ln Z(\lambda) &= \sum_{\substack{\lbrack\mathds{G}\rbrack\ \text{connected, rooted,} \\ \text{at least one } E_q>0}}
\frac{1}{2^{C(\mathds{G}-\mathcal{E}\suprho)+1} \sum_{q\in\mathscr{Q}} E_q} \; 
\mathscr{A}(\mathds{G}) \; , \crcr
G_2( \lambda ) &= \sum_{\lbrack\mathds{G}\rbrack\ \text{connected, rooted }}
\frac{1}{2^{C(\mathds{G}-\mathcal{E}\suprho)-1}} \; 
\mathscr{A}(\mathds{G}) \; ,
\end{split}
\end{equation}
with amplitude:
\begin{equation}\label{eq:ampliintro}
\mathscr{A}(\mathds{G}) = 2^{E\suprho(\mathds{G})}\prod_{q\in\mathscr{Q}}   \left(-\lambda_q\right)^{E_q(\mathds{G})} \prod_{c=1}^D \big((-1)^{\abs c} N_c\big)^{F_c(\mathds{G})} \; ,
\end{equation}
where \(E_q\), \(E\suprho\), \(F_c\), \(C(\mathds{G}-\mathcal{E}\suprho)\) are some combinatorial numbers associated to the multi-ribbon graph \(\mathds{G}\) (see Sec.\,\ref{sec: tensor-if} for the relevant definitions). 
\end{theorem}

\begin{proof}
The theorem follows from Eq. \eqref{eq: tensor-amplitude}, \eqref{eq: G2-tensor} and \eqref{eq: lnZ-tensor}.

\end{proof}

The crucial remark is that all the factors \(N_c\) come in the form \((-1)^{\abs c} N_c\), hence  each term is mapped into itself by exchanging $O(N_c) \leftrightarrow  Sp(N_c) $ and $N_c \leftrightarrow -N_c$.


\paragraph{Conclusion and Outlook.}\label{sec: conclusion} We list some comments on, and possible generalizations of, our result:
\begin{itemize}
\item In order to prove our main theorem we will use in this paper an intermediate field representation adapted to quartic interactions. It should however be possible to extend this result to more general interactions \cite{Lionni-IF}. 
\item While more general models with \(OSp(m,n)\) symmetry could be considered, the construction of super tensor actions is complicated because of the abundance of sign factors \cite{Sasakura}.
\item For $D=2$ (matrices), the contributions of ribbon graphs and their duals cancel exactly in the fermionic case (see Remark~\ref{cor: cancellations}). It would be interesting to understand similar cancellations in the graded tensor models. This should be related to Poincaré duality between lower dimensional colored subgraphs.
\item One should explore the implications of the $N\to -N$ duality for tensor field theories. The sign changes may generate new renormalization group fixed points, and the duality may not hold for all the physical properties \cite{LeClair}. Quantum mechanical models of order three tensors with \(Sp(N)\) symmetry have been studied in \cite{diversenumbers, Carrozza}.
\end{itemize}

\paragraph{Outline of the paper.} This paper is organized as follows. In Sec.\,\ref{sec: deftensor} the quartic graded tensor model is defined, the relation between directed edge colored graphs and quartic trace invariants is explained, and we collects some definitions and notations on ribbon graphs,
Sec.\,\ref{sec: matrix} deals in detail with the order 2 (matrix) case.
Sec.\,\ref{sec: tensor} continues with the general case of arbitrary order \(D\) tensors.
Appendix~\ref{ap: ribbon} contains the calculation of the sign of each ribbon graph amplitude and Appendix~\ref{ap: symmfac} gives details on the calculation of the symmetry factors of the Feynman graphs. 


\subsection*{Acknowledgments}

The authors would like to thank Dario Benedetti for comments and discussions at the early stages of this project.
This work has been supported by the European Research Council (ERC) under the European Union’s Horizon 2020 research and innovation program (grant agreement No818066) and by the Deutsche Forschungsgemeinschaft (DFG, German Research Foundation) under Germany’s Excellence Strategy EXC-2181/1 -
390900948 (the Heidelberg STRUCTURES Cluster of Excellence).


\section{Definitions}\label{sec: deftensor}

In this section we define the models we will be studying. We also give some standard definitions about ribbon graphs and combinatorial maps.

\subsection{The Real Quartic Graded Tensor Models}

The orthosymplectic super Lie group
\(OSp(m,n)\) is the isometry group of the canonical graded-symmetric bilinear form on the supervector space \(\mathds{R}^{m|n}\):
\begin{equation*}
	\eta\, :\ \mathds{R}^{m|n} \times \mathds{R}^{m|n} \longrightarrow \Lambda_\infty,
	\quad
	\left(\eta_{ij}\right) =
	\left(\begin{array}{c|c}
		\,\delta & 0 \\ \hline
		0 & \omega
	\end{array}\right) \;, 
\end{equation*}
where \(\Lambda_\infty\) is the graßmann algebra generated by an infinite number of anticommuting generators. \(\mathds{R}^{m|n}\) is a free module over \(\Lambda_\infty\) with \(m\) even (commuting) and \(n\) odd (anticommuting) basis vectors. Note that non-singularity of \(\eta\) demands that \(n\) is an even integer. For later comparison, \(m\) is also taken to be even.

Since we are only interested in \(O(N)\) and \(Sp(N)\), and not the whole \(OSp(m,n)\), we restrict to supervector spaces that are either purely odd or purely even, and thus have either \(O(N)\) or \(Sp(N)\) as their isometry group. This information can be encoded as a parity of the index color \(\abs c\in\{0,1\}\), with \(\abs c=0\) corresponding to orthogonal, and  \(\abs c=1\) to symplectic symmetry. The tensor components are commuting bosonic or anticommuting gra\ss mannian, depending on whether the number of indices with $\abs{c}=1$ is even or odd.
Suitable invariants are defined to construct the actions of the models.

\paragraph{Vector spaces.} Let $\mathscr{H}_c  =  \mathds{R}^{N_c|0}$ for $\abs{c}=0 $, respectively $ \mathscr{H}_c  =\mathds{R}^{0|N_c}$ for $\abs{c}=1  $ be a real supervector space of dimension \(N_c\) that is either purely even or purely odd and is endowed with a non-degenerate graded symmetric inner product 
$g^c\, : \mathscr{H}_c\times \mathscr{H}_c \rightarrow \Lambda_\infty$:
\[
g^c(u,v) = (-1)^{\abs{c}} g^c(v,u),\quad \forall u, v\in \mathscr{H}_c  \;,\qquad
g^c(\,\cdot\, ,v) = 0\ \Leftrightarrow \ v=0  \;. 
\]

In a standard basis \(g^c\) agrees with the standard symmetric or symplectic form, that is $g^c_{a_cb_c} = \delta_{a_cb_c}$ for $\abs{c}=0 $, respectively $ g^c_{a_cb_c} = \omega_{a_cb_c}$ for $\abs{c}=1 $. We denote $(g^c)^{a_cb_c}$ the matrix element of the inverse $(g^c)^{-1}$. The isometry group preserving \(g^c\) is either \(O(N_c)\) in the \(\abs{c}=0\) case or \(Sp(N_c)\) in the \(\abs{c}=1\) case, denoted collectively by 
$\pmb{O}_{c}(N_c) := \{ O_c\ |\ g^c_{a_cb_c} = O_{a_c}^{\ a'_c} O_{b_c}^{\ b'_c} g^c_{a'_c b'_c}  = ( O g^c O^T )_{a_cb_c}\}$.

\paragraph{Tensors.}
Tensors are \emph{even} elements of the tensor product space
$T\in \bigotimes_{c=1}^{D} \mathscr{H}_c$.
Choosing a basis  \(\{[\psi^c]_{a_c}\}_{a_c=1,\dots,N_c} \) in each \(\mathscr{H}_c\) and denoting the dual basis by \(\{[\psi^\vee_c]^{a_c}\}_{a_c=1,\dots,N_c} \), the components of a tensor are:
\[
T^{a_1\dots a_D} \equiv T([\psi^\vee_1]^{a_1},\dots, [\psi^\vee_D]^{a_D}),\quad T= \sum_{a_c=1,\dots,N_c,\,\forall c} T^{a_1\dots a_D}\ [\psi^1]_{a_1}\otimes \dots\otimes [\psi^D]_{a_D} \;.
\]

A generic tensor has no symmetry properties under permutation of its indices \(a^1,\dots,a^D\) hence the indices have a well defined position \(c\), called their \emph{color}. The set of colors is denoted \(\mathcal{D}=\{1,\dots, D\}\). We sometimes call the colors with $\abs{c}=0$ even and the ones with $\abs{c}=1$ odd. As the tensors are taken to be even elements of the tensor product space, the tensor components are bosonic (even) if the number of colors with $\abs{c}=1$ (i.\,e.\ odd colors) is even and fermionic (odd) otherwise: the Graßmann number $T^{a_1\dots a_D}$ has the same parity as $\sum_{c\in\mathcal{D}} \abs{c}$.

The tensors transform in the tensor product representation of several orthogonal and symplectic groups according to the type of the individual \(\mathscr{H}_c\)'s:
\[
	T^{a_1\dots a_D} \to \tensor{(O_1)}{^{a_1}_{b_1}} \dots \tensor{(O_D)}{^{a_D}_{b_D}}\ T^{b_1\dots b_D},
	\quad O_1\otimes\dots\otimes O_D\in\bigotimes_{c\in\mathcal{D}} \pmb{O}_c(N_c) \; .
\]

A tensor can be viewed as a multilinear map
$T\, :\ \bigotimes_{c\in\mathcal{C}}\mathscr{H}_c^\vee \rightarrow \bigotimes_{c\in\mathcal{D}\backslash\mathcal{C}}\mathscr{H}_c$ for any subset of colors \mbox{\(\mathcal{C}\subset\mathcal{D}\)}. 
As the inner product \(g^c\) induces an isomorphism between \(\mathscr{H}_c\) and its dual, denoting 
\mbox{$a_{\mathcal{C}} = (a_c,\ c\in\mathcal{C})$}, the matrix elements of this linear map in the tensor product basis
are  \mbox{\(T^{a_{\mathcal{D}\backslash\mathcal{C}}\, a_{\mathcal{C}}}\equiv T^{a_1\dots a_D}\)}.

\paragraph{Edge colored graphs.}
Invariant polynomials in the tensor components can be constructed by contracting the indices of color \(c\) with the inner product \(g^c\).
The unique quadratic invariant is:
\begin{equation*}
	g^{\otimes\mathcal{D}}(T,T) := T^{a_{\mathcal{D}}} T^{b_{\mathcal{D}}} \prod_{c\in\mathcal{D}} g^c_{a_c b_c} \; .
\end{equation*}
General \emph{trace invariants} are polynomials in the \(T^{a_{\mathcal{D}}}\)'s build by contracting pairs of indices of the same color. These invariants form an algebraic complete set for all invariant polynomials and admit a straightforward graphical representation as \emph{edge colored graphs}. 

\begin{definition}[Edge Colored Graphs \cite{gurau}]\label{def: coloredgraph}
A closed edge \(D\)-colored graph is a graph \mbox{\(\mathcal{B}=(\mathcal{V}(\mathcal{B}),\mathcal{E}(\mathcal{B}))\)} with vertex set \(\mathcal{V}(\mathcal{B})\) and edge set \(\mathcal{E}(\mathcal{B})\) such that:
\begin{itemize}
\item The edge set is partitioned into \(D\) disjoint subsets \mbox{\(\mathcal{E}(\mathcal{B})=\bigsqcup_{c=1}^{D} \mathcal{E}^c(\mathcal{B})\)}, where \linebreak[5]\mbox{\(\mathcal{E}^c(\mathcal{B})\ni e^c=(v,w)\)}, \(v,w\in\mathcal{V}(\mathcal{B})\), is the subset of edges of color \(c\).
\item All vertices are \(D\)-valent with all the  edges incident to a vertex having distinct colors.
\end{itemize}
\end{definition}
In order to incorporate the odd colors appropriately, one needs to consider \emph{directed graphs}, that is graphs with an additional arrow for every edge (see Fig.~\ref{fig: quartic inv} for an example). Two graphs which are identical up to reorienting one edge of an odd color represent the same invariant up to a global $-$ sign. We will fix the global sign in the case of quartic invariants below.

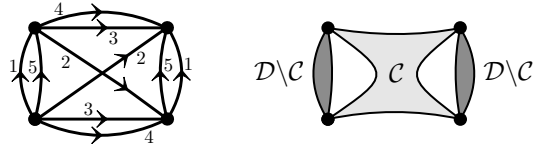
\begin{figure}[ht]
	\centering
	\begin{tikzpicture}[scale=0.875,font=\small]
		\node {\includegraphics[scale=3.5]{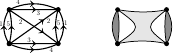}};
		\node at (2.32,0) {\(\mathcal{C}\)};
		\node at (0.55,0) {\(\mathcal{D}\backslash\mathcal{C}\)};
		\node at (4,0) {\(\mathcal{D}\backslash\mathcal{C}\)};
	\end{tikzpicture}
	\caption{\emph{Left:} Quartic 5-colored graph. \emph{Right:} Schematic representation of a general quartic invariant.}
	\label{fig: quartic inv}
\end{figure}

\paragraph{Quartic invariants.} 
Quartic invariants are represented by \(D\)-colored graphs with four vertices (see Fig.~\ref{fig: quartic inv}) and directed edges. Due to the sign ambiguity induced by reversing the edges corresponding to the odd colors, we need to give a prescription to fix the global sign of an invariant. Every directed quartic $D$-colored graph can be canonically oriented as follows (see again Fig.~\ref{fig: quartic inv}):
\begin{itemize}
 \item the color $1$ edges give a pairing of the vertices. We denote $a^1$ and $b^1$ the source vertices of the oriented edges $1$ and $a^2$ and $b^2$ their targets.
 \item we orient all the edges that connect $(a^1,a^2)$ respectively $(b^1,b^2)$ parallel to the edges $1$. We denote their colors $c\in \mathcal{D}\setminus\mathcal{C}$.
 \item all the edges of colors $c\in \mathcal{C}$ connect the $a$ pair with the $b$ pair. We orient all of them from the $a$ pair to the $b$ pair. These edges fall into two classes
    \begin{itemize}
     \item either they connect $a^1$ with $b^1$ and $a^2$ with $b^2$ in which case we say they run in the \emph{parallel} channel
     \item or they connect $a^1$ with $b^2$ and $a^2$ with $b^1$ in which case we say they run in the \emph{cross} channel.
    \end{itemize}
\end{itemize}
A canonically oriented graph is indexed by a subset of colors \(\mathcal{C}\subset\mathcal{D}\), \(1\notin\mathcal{C}\) and permutations of two elements \(\pi^c\in\mathfrak{S}_2=\{ \mathrm{id},(12)\}, \, c\in \mathcal{C}\). The associated invariant is:
\begin{equation}\label{eq: quarticinv}
\begin{split}
I(T) & = \sum_{a_{\mathcal{D}}^1, a_{\mathcal{D}}^2, b_{\mathcal{D}}^1, b_{\mathcal{D}}^2}
\Big(T^{a_{\mathcal{D}}^1}T^{a_{\mathcal{D}}^2} \prod_{c\in\mathcal{D}\backslash\mathcal{C}} g^c_{a_c^1 a_c^2} \Big)
\Big(T^{b_{\mathcal{D}}^1}T^{b_{\mathcal{D}}^2} \prod_{c\in\mathcal{D}\backslash\mathcal{C}} g^c_{b_c^1 b_c^2} \Big)\crcr
& \hspace{5.4em} \cdot\Big(\prod_{c\in\mathcal{C}} \big(-\sgn(\pi^c)\big)^{\abs c} g^c_{a_c^1 b_c^{\pi^c(1)}} g^c_{a_c^2 b_c^{\pi^c(2)}}\Big) \\
& = \sum_{a_{\mathcal{C}}^1, a_{\mathcal{C}}^2, b_{\mathcal{C}}^1, b_{\mathcal{C}}^2}
\big(g^{\otimes\mathcal{D}\backslash\mathcal{C}}(T,T)\big)^{a_{\mathcal{C}}^1 a_{\mathcal{C}}^2}
K_{a_{\mathcal{C}}^1 a_{\mathcal{C}}^2, b_{\mathcal{C}}^1 b_{\mathcal{C}}^2}
\big(g^{\otimes\mathcal{D}\backslash\mathcal{C}}(T,T)\big)^{b_{\mathcal{C}}^1 b_{\mathcal{C}}^2} \; ,
\end{split}
	\end{equation}
where we introduced the shorthand notation $K$ for the contractions of the indices transmitted between the pairs. Note that this is invariant by exchanging the $b$ vertices and that $ \big(-\sgn(\pi^c)\big)^{\abs c}$ is the signature of the permutation $(a^1a^2)(b^1b^2)$ to $(a^1 b^{\pi(1)} ) (a^{2}b^{\pi(2)})$ for the odd colors.

\begin{lem}\label{lem:invcount}
There are \(\frac{1+3^{D-1}}{2}\) different quartic trace invariants (see Fig.~\ref{fig:rk3inv} for the $D=3$ case).
\end{lem}

\begin{proof}
 There is only one invariant corresponding to 
 $\mathcal{C}=\emptyset$. If $\mathcal{C}$ has $q$ elements, there are $2^q$ choices for the channels and an overall $1/2$ for the relabeling of the $b$ vertices. Thus the total number of invariants is:
 \[
  1 + \frac{1}{2} \sum_{q=1}^{D-1} \binom{D}{q} 2^{q}  
   = \frac{3^{D-1} +1 }{2} \;.
 \]
\end{proof}

\begin{figure}[ht]
	\centering
	\begin{tikzpicture}[scale=0.875,font=\small]
		\node {\includegraphics[scale=3.5]{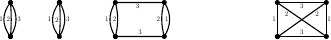}};
		\node at (2,0) {\(+ 2\left(\parbox{1.8cm}{\centering color permutations}\right)\)};
	\end{tikzpicture}
	\caption{The 5 quartic invariants at order 3 know as double trace, pillow and tetrahedron.}
	\label{fig:rk3inv}
\end{figure}

Denote the set of distinct quartic \(D\)-colored graphs and the associate trace invariants by \(\mathscr{Q}\ni q\) and \( I^{q}(T)\) respectively.
\begin{definition}[Real Quartic Graded Tensor Model]\label{def: tensormodel} The \emph{real quartic ``graded'' tensor model} is the measure:
\begin{gather*}
d\mu[T]=e^{-S[T]}\ [dT], \quad [dT]=\prod_{a_{\mathcal{D}}} dT^{a_1\dots a_D} \cdot
\begin{cases}
\frac{1}{(2\pi)^{  \prod_c N_c/2 } }, & \sum_{c=1}^{D} \abs{c} = 0 \mod 2\\
1, & \sum_{c=1}^{D} \abs{c} = 1 \mod 2
\end{cases} \; ,
\\ \text{with}\quad
S[T]= \frac{1}{2}g^{\otimes\mathcal{D}}(T,T) + \sum_{q\in\mathscr{Q}} \frac{\lambda_q}{4} I^{q}(T) \; ,
\end{gather*}
where the normalization is such that \(\int d\mu[T] =1\) for \(\lambda_q=0\ \forall q\in\mathscr{Q}\).

\end{definition}

\paragraph{Convergence issues.} Throughout this paper we treat the measures $d\mu[T]$ as perturbed Gaussian measures. As such we do not concern ourselves with the convergence of the various tensor and matrix integrals. The integrals are always convergent if $T$ is fermionic. If $T$ is bosonic, the integrals converge if $|c|=0$ for all $c$, but not necessarily in the other cases. As we treat the Gaussian integrals as generating functions of graphs, we will not worry about such issues.

\subsection{Ribbon Graphs and Combinatorial Maps}\label{sec: ribbon}

As ribbon graphs \cite{embeddedgraphs, grosstucker} and combinatorial maps play a significant role in the derivation of our results, we review here some of their properties.

\begin{figure}[ht]
\centering
\includegraphics[scale=5]{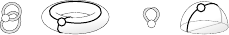}
\caption{Ribbon graphs, which we denote  \(\mathcal{G}_{T^2}\) and \(\mathcal{G}_{\mathds{RP}^2}\), and their cellular embeddings. The rightmost surface is the hemisphere representation of the real projective plane where opposite points along the equator are identified.}
\label{fig: ap-ribbon}
\end{figure}

Ribbon graphs, see Fig.~\ref{fig: ap-ribbon} for some examples, are cellularly embedded graphs on topological surfaces, and thus can be viewed as 2-cell-complexes. Due to the embedding, each vertex carries an orientation and the order of edges around a vertex is fixed. A vertex can be re-embedded with the opposite orientation: this amounts to reversing the order of the incident edges and giving them a twist, see Fig.~\ref{fig: reembedding}.
	
\begin{definition}[Ribbon Graph \cite{embeddedgraphs}]\label{def: ribbon}
A \emph{ribbon graph} \(\mathcal{G}=( \mathcal{V}(\mathcal{G}), \mathcal{E}(\mathcal{G}) )\) is a (possibly non-orientable) surface with boundary, represented as the union of two sets of topological discs, a set of \emph{vertices} \(\mathcal{V}(\mathcal{G})\), and a set of \emph{edges} \(\mathcal{E}(\mathcal{G})\), such that:
\begin{enumerate}
\item The vertices and edges intersect in disjoint line segments.
\item Each such line segment lies on the boundary of precisely one vertex and precisely one edge.
\item Every edge contains exactly two such line segments.
\end{enumerate}
The boundary components of \(\mathcal{G}\) are called \emph{faces}. The two disjoint boundary segments of an edge that are not connected to a vertex (i.e. the two sides of the edge) are called \emph{strands}.  We denote the set of faces of $\mathcal{G}$ by $\mathcal{F}(\mathcal{G})$. A ribbon graded becomes a 2 dimensional CW complex by sewing two dimensional patches along its faces.

The numbers of vertices, edges and faces of $\mathcal{G}$ are denoted by \(V(\mathcal{G})\), \(E(\mathcal{G})\) and \(F(\mathcal{G})\), respectively.
\end{definition}

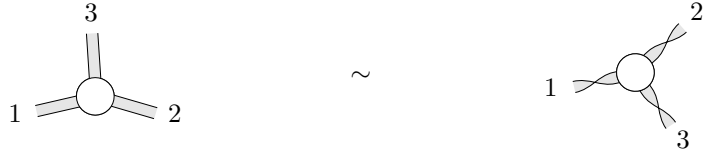
\begin{figure}[ht]
		\centering
		\begin{tikzpicture}[scale=1, font=\small]
			\node at (0,0) {\includegraphics[scale=6]{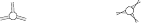}};
			\node at (-4.55,-.5) {1};
			\node at (-3.55,.84) {3};
			\node at (-2.45,-.5) {2};
			\node at (4.42,.85) {2};
			\node at (4.24,-.843) {3};
			\node at (2.5,-.15) {1};
			\node at (0,0) {\(\sim\)};
		\end{tikzpicture}
		\caption{Re-embedding a vertex: the order of halfedges is reversed and they gain additional twists. This is an equivalence relation of ribbon graphs.}
		\label{fig: reembedding}
\end{figure}
    
Several remarks are in order:
\begin{itemize}
\item The strands of an edge can run parallel, in which case the edge is called untwisted, or cross, in which case the edge is called twisted.
\item If \(V(\mathcal{G})=1\) the graph is called a \emph{rosette graph}. A rosette graph with only one face is called a \emph{superrosette graph}.
\item A \emph{self-loop} in \(\mathcal{G}\) is an edge \(e=\{h_v,h'_v\} \in\mathcal{E}(\mathcal{G})\) connected to just one vertex \(v\in\mathcal{V}(\mathcal{G})\). 
A \emph{simple self-loop} is a self-loop such that its halfedges are direct neighbors in the cyclic ordering around \(v\), thus \(v\) has a corner of the form \( (h_v,h'_v)\). If \(e\) is (un-)twisted the simple self-loop is called likewise.
\item We denote the ribbon graph consisting in only one vertex with no edge by \(\mathcal{G}_{\kern-.1em\mathlarger{{\circ}}}\). By definition this graph has one face.
We denote the ribbon graph with one vertex and one twisted self-loop edge by \(\mathcal{G}_{\mathds{RP}^2}\) and the ribbon graph with one vertex, two untwisted self-loop edges but no simple self-loop by \(\mathcal{G}_{T^2}\)\footnote{As their names suggest, these graphs can be cellularly embedded into \(\mathds{RP}^2\) or \(T^2\), respectively.}. The last two graphs are depicted in Fig.~\ref{fig: ap-ribbon}. As a topological surface with boundary \(\mathcal{G}_{\mathds{RP}^2}\) is homeomorphic to a Möbius strip.
\end{itemize}

Every ribbon graph has a dual ribbon graph with the same number of edges, but with the roles of the vertices and the faces interchanged.
\begin{definition}[Dual Ribbon Graph
\cite{embeddedgraphs}]\label{def: dualgraph}
		Let \(\mathcal{G}\) be a ribbon graph. The \emph{dual} ribbon graph \(\mathcal{G}^\ast\) is obtained by sewing discs along the faces of \(\mathcal{G}\) and deleting the original vertex discs of \(\mathcal{G}\). The new discs make up the dual vertex set \(\mathcal{V}(\mathcal{G}^\ast)\), and the new boundary components created by the deletion are the faces  of \(\mathcal{G}^\ast\). See Fig.~\ref{fig: dualgraph} for an illustration. 
	\end{definition}

	\begin{figure}[ht]
		\centering
		\includegraphics[scale=5]{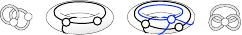}
		\caption{The dual graph. }
		\label{fig: dualgraph}
	\end{figure}
	
Besides ribbon graphs, we will encounter combinatorial maps below.

\begin{definition}[Combinatorial Map]\label{def: c-map}

A \emph{combinatorial map} \(\mathcal{M}=(\mathcal{S},\pi,\alpha)\) is a finite set \(\mathcal{S}\) of halfedges (or darts) of even cardinality, together with a couple of permutations \((\pi, \alpha)\) on \(\mathcal{S}\), where \(\alpha\) is an involution with no fixed points (a ``pairing'' of halfedges).

\(\mathcal{M}\) is called connected if the group freely generated by \(\pi\) and \(\alpha\) acts transitively on \(\mathcal{S}\).
The dual of \(\mathcal{M}\) is the combinatorial map \(\mathcal{M}^\ast=(\mathcal{S},\alpha\circ\pi,\alpha)\).
\end{definition}

Combinatorial maps can be represented as graphs embedded in orientable surfaces. The cycles of \(\pi\) represent vertices with a cyclic order of their halfedges (chosen to be counter-clockwise), and \(\alpha\) encodes pairings of halfedges into edges. The faces of a combinatorial map are the cycles of the permutation \(\alpha\circ\pi\). In the dual combinatorial map, the role of vertices and faces is reversed.

The definition of combinatorial maps and ribbon graphs can be extended to include a second kind of edges. 

\begin{definition}[Combinatorial Map with \(\varrho\)-Edges]\label{def: c-map-rho}\ A \emph{combinatorial map with \(\varrho\)-edges} \linebreak[4] \mbox{\(\mathcal{M}\suprho=(\mathcal{S}\sqcup\mathcal{S}\suprho,\pi, \alpha,\alpha\suprho)\)} is a finite set \(\bar{\mathcal{S}}=\mathcal{S}\sqcup\mathcal{S}\suprho\) that is the disjoint union of two sets of halfedges, both of even cardinality, together with a triple of permutations \((\pi, \alpha,\alpha\suprho)\) on \(\bar{\mathcal{S}}\). \(\alpha\) and \(\alpha\suprho\) are fixed-point free involutions on \(\mathcal{S}\) and \(\mathcal{S}\suprho\) respectively, and extended to the whole of \(\bar{\mathcal{S}}\) by setting \(\alpha(h)=h\ \forall h\in\mathcal{S}\suprho\) and analogous for \(\alpha\suprho\).
		
The cycles of \(\alpha\suprho\) are pairs of halfedges in \(\mathcal{S}\suprho\) which we will call \(\varrho\)-edges. 
\(\mathcal{M}\suprho\) is connected if the group freely generated by \(\pi\), \(\alpha\) and \(\alpha\suprho\) acts transitively on \(\bar{\mathcal{S}}\).
The cycles of \(\pi\) are the vertices and the cycles of \(\pi\circ\alpha\) are the faces of \(\mathcal{M}\suprho\).
The dual map is defined by changing the role of vertices and faces but not touching the \(\varrho\)-edges \(\mathcal{M}\suprho^\ast =(\mathcal{S}\sqcup\mathcal{S}\suprho,\alpha\circ\pi, \alpha,\alpha\suprho)\).

Deleting all the \(\varrho\)-edges one obtains an ordinary combinatorial map.
\end{definition}

Ribbon graphs can be obtained from combinatorial maps by replacing their edges with twisted or untwisted ribbon edges. The same holds true for combinatorial maps with \(\varrho\)-edges and \emph{ribbon graphs with \(\varrho\)-edges}.

\begin{definition}[Ribbon Graph with \(\varrho\)-Edges]\label{def: ribbon-rho}
A \emph{ribbon graph with \(\varrho\)-edges} \mbox{\(\mathcal{G}\suprho=( \mathcal{V}, \mathcal{E},\mathcal{E}\suprho )\)} is a ribbon graph \(\mathcal{G}=( \mathcal{V}, \mathcal{E})\), together with a set of line segments \(\mathcal{E}\suprho\), called \(\varrho\)-edges, such that their endpoints are connected to the corners of the ribbon graph. \(\mathcal{G}\suprho\) is called connected if it is connected as a topological space. 
The notions of faces, corners and edges of \(\mathcal{G}\suprho\) refer to the ones of the ribbon graph \(\mathcal{G}=\mathcal{G}\suprho-\mathcal{E}\suprho\), that is obtained by deleting the \(\varrho\)-edges.\footnotemark

This dual of a ribbon graph with $\varrho$-edges is obtained by performing the partial dual \cite{chmutov,embeddedgraphs,twistedduality} with respect to the ribbon edges.
This is the dual of the underlying ribbon graph obtained by ignoring the \(\varrho\)-edges, where we keep track of the corners to which the \(\varrho\)-edges are hooked. 
\end{definition}
\footnotetext{Ribbon graphs with \(\varrho\)-edges are-embedded in Nodal Surfaces, that is Riemann surfaces glued at marked points. The ribbon graphs encode closed topological surfaces and by identifying points that are connected by an \(\varrho\)-edge, a gluing prescription is given.}

\section{Matrix Models}\label{sec: matrix}

We first deal with the case of matrices (order $D=2$ tensors) in Def.~\ref{def: tensormodel}. In particular:
\begin{equation*}
	M^{a_1a_2}M^{b_1b_2} = (-1)^{\abs{1}+\abs{2}} M^{b_1b_2}M^{a_1a_2} \; ,
\end{equation*}
i.e. the models with mixed symmetry are fermionic.
We show that, for each ribbon graph in the perturbative expansion of the free energy and the two point function of the model, changing one (or both) of the symmetry group factors in the \(O(N_1)\otimes O(N_2)\)-model from \(O(N)\) to \(Sp(N)\) amounts to changing the sign accompanying the corresponding \(N\) factor. 

Complex random matrix models in the intermediate field representation have been studied in \cite{guraumatrix}.  The sign changes between the \mbox{\(O(N_1)\otimes O(N_2)\)} and the \mbox{\(Sp(N_1)\otimes Sp(N_2)\)} model has also been studied in \cite{goe-gse} by different methods. 

Denoting with superscript $T$ the transpose, the action of the real quartic graded matrix model writes:\footnote{In the pure \(Sp(N)\) case with \(g^1=g^2=\omega\), the convergence of \eqref{eq: supermatrixaction} is not clear, since the quadratic part has negative modes.}
\begin{align}\label{eq: supermatrixaction}
	S[M] 	=&
	\frac{1}{2} M^{a_1a_2}g^1_{a_1 b_1} g^2_{a_2 b_2} M^{b_1 b_2} + \frac{\kappa}{4} \Big(M^{a_1a_2}g^1_{a_1 b_1} g^2_{a_2 b_2} M^{b_1 b_2}\Big)^2 \nonumber \\
	&+ \frac{\lambda}{4}(-1)^{\abs{2}}\, ( M^{a^1_1a^1_2}g^1_{a^1_1 a^2_1}M^{a^2_1a^2_2} ) g^2_{a^1_2 b^1_2 } g^2_{ a^2_2 b^2_2 } 
	 ( M^{b^1_1b^1_2}g^1_{b^1_1 b^2_1}M^{b^2_1b^2_2} )  \; ,\nonumber \\
	=&\frac{1}{2}\tr\big[ Mg^2 M^T (g^1)^T \big]
	+ \frac{\kappa}{4}\big( \tr\big[ Mg^2 M^T (g^1)^T \big]\big)^2
	+ \frac{\lambda}{4}\,(-1)^{\abs 1}\, \tr\big[\big(Mg^2 M^T g^1\big)^2\big] \;,
\end{align}
where we note that trace is 
$\tr[A]=A^{a}_{\ a} = A^{ab} g_{ba}$.
This action is invariant under the transformation
\(	M\to O_1 X O_2^{T} \) with \(O_1\in \pmb{O}_{1}(N_1),\ O_2\in \pmb{O}_{2}(N_2) \). The three terms in Eq.~\eqref{eq: supermatrixaction}
can be represented by 2-colored graphs or alternatively ribbon graphs, as depicted in Fig.~\ref{fig: oomatrix}.

\begin{figure}[ht]
	\centering
	\includegraphics[scale=3]{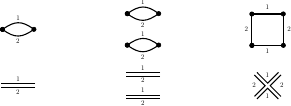}
	\caption{Graphical representation of the matrix model invariants up to quartic order.}
	\label{fig: oomatrix}
\end{figure}

Whereas all terms in the action of the \(O(N_1)\otimes O(N_2)\)-model are positive for \(\kappa, \lambda\in\mathds{R}_+ \), in the \(Sp(N_1)\otimes Sp(N_2)\)-model, this is only true for the $\kappa$ term: the quadratic and the $\lambda$ terms are in general indefinite.

\subsection{Intermediate Field Representation}

The intermediate field (Hubbard-Stratonovich) representation is obtained by introducing an auxiliary field per quartic interaction and integrating out the original field. To be precise we use:
\begin{align}\label{eq: ifreplacement}
	&\exp\Big\lbrace - \frac{\kappa}{4} \big(\tr\big[Mg^2 M^{T} (g^1)^T\big]\big)^{2} -\frac{\lambda}{4} (-1)^{\abs{1}}\tr\big[ Mg^2M^Tg^1Mg^2M^Tg^1 \big] \Big\rbrace
	\\ 
	&=
	\Big[ e^{ \frac{1}{2}\left(\frac{\del}{\del\sigma} P \frac{\del}{\del\sigma}\right)}
	e^{ \frac{1}{2}\frac{\del}{\del\varrho}\frac{\del}{\del\varrho}} \;
	\exp\Big\lbrace  - \imath \sqrt{\frac{\kappa}{2}}\, \tr\big[Mg^2 M^{T} (g^1)^T\big] \rho - \imath \sqrt{ \frac{\lambda}{2}}\tr\big[ Mg^2M^T (g^1 \sigma g^1)^T \big] \Big\rbrace
	\Big]_{\varrho, \sigma=0} \; , \nonumber
\end{align}
where \(\varrho\) is a real commuting (bosonic) scalar field and \(\sigma=(-1)^{\abs 1}\sigma^{T}\) is a (bosonic) real graded-symmetric matrix and we introduce the shorthand notation:
\begin{align*}
	\Big(\frac{\del}{\del\sigma} P\frac{\del}{\del\sigma}\Big) := \frac{\del}{\del\sigma^{ab}}P^{ab,dc}\frac{\del}{\del\sigma^{cd}} \;,\qquad	
P^{ab,dc} = \frac{1}{2} \left( g_1^{ad}g_1^{bc} + (-1)^{\abs 1}\, g_1^{ac}g_1^{bd} \right)  \;,
\end{align*}
with \(P\) the (anti-)symmetric projector, taking into account the symmetry of the \(\sigma\) field. Note that $g^1 Mg^2 M^T g^1$ has the same graded symmetry as $\sigma$.

Equation \eqref{eq: ifreplacement} is just a Gaussian integral over the intermediate fields $\varrho$ and $\sigma$. We favor here the notation of the Gaussian integral as a differential operator (see for instance
\cite{brydges2007lectures}) for two reasons. First, the Gaussian integral is formal in some cases (that is the covariance is not necessarily positively defined). Second, in this form the perturbative expansion of the Gaussian integral is straightforward.
In order to prove \eqref{eq: ifreplacement} we expand the exponentials and commute the sum and the derivatives:
\[
 \begin{split}
&\Big[ e^{ \frac{1}{2}\left(\frac{\del}{\del\sigma} P \frac{\del}{\del\sigma}\right)} e^{ \frac{1}{2}\frac{\del}{\del\varrho}\frac{\del}{\del\varrho}}
\sum_{n,p=0}^{\infty}  \tfrac{ (- \tfrac{\kappa}{2} )^n  (- \tfrac{\lambda}{2} )^p }{(2n)!(2p)!}  \big(\tr\big[Mg^2 M^{T} (g^1)^T\big] \rho\big)^{2n} \big(\tr\big[ Mg^2M^T (g^1 \sigma g^1)^T \big] \big)^{2p} 
	\Big]_{\varrho, \sigma=0}
\crcr
& =\Big[ \sum_{n,p=0}^{\infty} \frac{(- \tfrac{\kappa}{2} )^n  (- \tfrac{\lambda}{2} )^p } {2^n n! 2^p p! (2n)!(2p)!}
(\tfrac{\del}{\del\varrho}\tfrac{\del}{\del\varrho})^n
\big(\tr\big[Mg^2 M^{T} (g^1)^T\big] \rho\big)^{2n}
(\tfrac{\del}{\del\sigma}P \tfrac{\del}{\del\sigma} )^p \big(\tr\big[ Mg^2M^T (g^1 \sigma g^1)^T \big] \big)^{2p} 
	\Big]_{\varrho, \sigma=0}
\crcr
&	= \Big[ \sum_{n,p=0}^{\infty} \frac{(- \tfrac{\kappa}{4} )^n  (- \tfrac{\lambda}{4} )^p }{n!p!} \big(\tr\big[Mg^2 M^{T} (g^1)^T\big]\big)^{2n}
\big((-1)^{\abs{1}}\tr\big[ Mg^2M^Tg^1Mg^2M^Tg^1\big]\big)^{p} 
\Big]_{\varrho, \sigma=0} \; ,
 \end{split}
\]
where we used
\( [g^1 Mg^2M^T g^1]_{ab} g^{ad}g^{bc}
[g^1 Mg^2M^T g^1]_{cd}    = (-1)^{|1|}  \tr[g^1 Mg^2M^T g^1 Mg^2M^T g^1]\). The partition function now reads:
\[
\begin{split}
Z(\kappa,\lambda) = & \int [dM]\, e^{-\frac{1}{2}\tr [Mg^2M^T(g^1)^T]} \crcr
& \quad \Big[ e^{ \frac{1}{2}\left(\frac{\del}{\del\sigma}P \frac{\del}{\del\sigma}\right)}
e^{ \frac{1}{2}\frac{\del}{\del\varrho}\frac{\del}{\del\varrho}}
e^{  - \imath \sqrt{\frac{\kappa}{2}}\, \tr[Mg^2 M^{T} (g^1)^T] \rho - \imath \sqrt{ \frac{\lambda}{2}}\tr[ Mg^2M^T (g^1 \sigma g^1)^T ] }
\Big]_{\substack{\varrho,\sigma=0}} \; ,
\end{split}
\]
and all the terms containing $M$ can be collected in a quadratic form using $\tr(MAM^TB^T) = M(B \otimes A)M$. The exponent writes 
$ - \frac{1}{2} M(R^{-1}\otimes g^2)M $
with $R$ the \emph{resolvent operator}:
\[
\tensor{[ R^{-1}(\kappa,\lambda) ]}{^{a}_{b}}  =
	(1+\imath\sqrt{2\kappa} \, \varrho) \tensor{\delta}{^a_{b}} +\imath\sqrt{2\lambda} \; \tensor{(\sigma g^1)}{^{a}_{b}} \;. 
\]
As the resolvent and its inverse are operators we write them with a covariant and a contravariant index. These indices are lowered with $g^1$ and raised with 
$(g^1)^{-1}$.

Commuting the integral and derivative operators, the \(M\) integral is gau\ss ian and can be performed leading to the \emph{intermediate field representation}:
\begin{align}\label{eq: intermedrep2}
	Z(\kappa,\lambda) = \Big[ e^{ \frac{1}{2}\left(\frac{\del}{\del\sigma}P \frac{\del}{\del\sigma}\right)}	e^{ \frac{1}{2}\frac{\del}{\del\varrho}\frac{\del}{\del\varrho}}  \;
	e^{ (-1)^{\abs 1 + \abs 2}\frac{N_2}{2} \tr\ln R(\kappa,\lambda)}
	\Big]_{\varrho, \sigma=0}  \; .
\end{align}
\(N_2\) is now an explicit parameter in the integral, while \(N_1\) is hidden in the remaining traces.
The sign \((-1)^{\abs 1 + \abs 2}\) tracks the bosonic/fermionic character of the original matrix. The sign 
\((-1)^{\abs 1} \) tracks the symmetry of the intermediate matrix field \(\sigma\) (which agrees with that of \(g^1\)). Both indices of $\sigma$ have color $1$ which reflects the fact that $\sigma$ transforms in the (anti-)symmetric tensor representation of \(\pmb{O}_1(N_1)\) that is 
\(	\sigma \to O_1 \sigma O_1^{T} \) for \( O_1\in \pmb{O}_1(N_1) \).
This is to be contrasted with the field \(M\) which transforms in the tensor product of the fundamental representations of \( \pmb{O}_1(N_1)\) and \(\pmb{O}_2(N_2)\). 

\subsection{Perturbative Expansion}\label{sec: pert exp and graphical}

The perturbative expansion of $Z$ is obtained by Taylor expanding the interaction:
\[
	Z(\kappa,\lambda) = \Big[ e^{ \frac{1}{2}\left(\frac{\del}{\del\sigma}P \frac{\del}{\del\sigma}\right)}
	e^{ \frac{1}{2}\frac{\del}{\del\varrho}\frac{\del}{\del\varrho}}
	\sum_{V=0}^{\infty} \frac{1}{V!}\Big(\frac{(-1)^{\abs 1 +\abs 2} N_2}{2}\, \tr\ln{R}(\kappa,\lambda;\varrho,\sigma)\Big)^V \Big]_{\varrho, \sigma=0}  \;,
\]
and commuting the gau\ss ian integration with the sum.
Note that $R$ denotes the resolvent operator hence it naturally has a covariant and a contravariant index. Taking into account that:
\[
\ln R = - \sum_{p\ge 1} \frac{(-1)^{p+1}} {p} \left(  
  \imath\sqrt{2\kappa} \, \varrho +\imath\sqrt{2\lambda} \; \sigma g^1  \right)^p \;, \qquad R = \sum_{p\ge 0} (-1)^p   \left(  
  \imath\sqrt{2\kappa} \, \varrho +\imath\sqrt{2\lambda} \; \sigma g^1  \right)^p \;,
\]
the derivatives of the resolvent and its logarithm are:
\[
\begin{split}
& \frac{\del}{\del\sigma^{ab}} \tr\ln{R} = - \imath\sqrt{2\lambda} P_{ab,cd} {R}^{dc} \;,\qquad
\frac{\del}{\del\sigma^{ab}} {R}^{cd} = -\imath\sqrt{2\lambda} P_{ab,ef}{R}^{ce}{R}^{fd}  \;, \crcr
& \frac{\del}{\del\varrho} \tr\ln{R} = - \imath\sqrt{2\kappa}\, \tr[R] \;, \qquad
 \frac{\del}{\del\varrho} {R}^{ab} = - \imath\sqrt{2\kappa}\, {R}^{ac}g^1_{cd}{R}^{db} = - \imath\sqrt{2\kappa} (R^2)^{ab}\; ,
\end{split}
\]
where $R^2$ denotes the square of the operator $R$.

Each term in the perturbative series can be represented as a ribbon graph with \(\varrho\)-edges (see Sec.\,\ref{sec: ribbon}) as depicted in Fig.~\ref{fig: matrix-arrows}:
\begin{itemize}
 \item we represent each \(\tr\ln{R}\) as a disk with boundary oriented counterclockwise. 

 \item the \(\sigma\) derivatives create ribbon halfedges representing the free indices of \(R\). The first derivative acting on a vertex creates a halfedge and an \(R\) associated to the \emph{corner} (region between two consecutive halfedges) of the vertex. Subsequent derivatives split the existing corners creating new \(R\)'s.
 
 The indices $ab$ of the resolvent $R^{ab}$ are associated to the ends of the corner: $a$ for the source and $b$ for the target in the sense of the arrow.
  
\item the ribbon halfedges are connected into ribbon edges corresponding to the projectors $P$ inside the 
\(
\frac{\del}{\del\sigma^{ab}}P^{ab,dc}\frac{\del}{\del\sigma^{cd}} \) operators. The edges have an orientation represented by arrows on the strands bounding an edge: corresponding to $P^{ab,dc}$ we orient the strands from $(ab)$ to $(dc)$. Note that: 
\[
\begin{split}
& ( \partial_{\sigma^{ab}}R^{pq} ) P^{ab,dc} (\partial_{\sigma^{cd} } R^{ef} ) = (-2\lambda) 
 (R^{pa'} R^{b'q}) P_{ ab , a'b'}P^{ab,dc} 
 P_{ cd ,c'd' }(R^{ec'} R^{d'f}) \crcr
 &  = (-\lambda) (R^{pa'} R^{b'q}) 
   2 P_{a'b'}^{\ \ \ d c} P_{ cd ,c'd' } 
 (R^{ec'} R^{d'f}) =  (-\lambda) (R^{pa'} R^{b'q}) 
 \; 2 P_{a'b'; d'c'}
  (R^{ec'} R^{d'f}) \; .
 \end{split}
\]

The projector generates two terms. The first one 
$g^1_{a'd'}g^1_{b'c'}$, corresponds to an edge with parallel strands. The second one $(-1)^{|1|}g^1_{a'c'}g^1_{b'd'}$ corresponds to a twisted edge.

\item a $ \varrho$ derivative splits corner of a vertex also, but connects these two halves by a \(g^1\). We represent this by a new type of halfedge, called \(\varrho\)-halfedge. The \(\varrho\)-halfedge 
are connected into \(\varrho\)-edges corresponding to the  $\frac{\del}{\del\varrho}\frac{\del}{\del\varrho}$ operators. We represent these edges as dashed lines.

In the end all intermediate fields are set to zero thus the resolvents are set to the identity \(R=\mathds{1}\). A corner that has been split by \(\varrho\)-halfedges behaves like a single ordinary corner of a ribbon graph: for this reason corner will always refer to the region between two ribbon halfedges only. 

\end{itemize}

\begin{figure}[ht]
	\centering
	\begin{tikzpicture}[font=\scriptsize]
		\node {\includegraphics[scale=5]{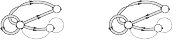}};
		\node at (0.04,0.05) {\(= \big((-1)^{\abs 1}\big)^5\)};
	\end{tikzpicture}
	\caption{\emph{Left:} A ribbon graph with \(\varrho\)-edges in the priori orientation: corners counter-clockwise and strands parallel. \emph{Right:} Coherent orientation of arrows along every face. Five arrows had to be reoriented.}
	\label{fig: matrix-arrows}
\end{figure}

Ignoring the twisting of the edges, a ribbon graph is a combinatorial map with \(\varrho\)-edges \(\mathcal{M}\suprho\). We denote by \(h_v\) the ribbon-halfedges of the vertex \(v\), each of which comes equipped with a pair of indices \((b_{h_v}, a_{h_v})\): $b$ is the target of an arrow and $a$ the source of another one. If \(h_v\) and \(h'_v\) are two neighboring ribbon-halfedges with \(h_v < h'_v\) in the cyclic order around \(v\), the corner between them is denoted by \((h_v, h'_v)\). An ribbon-edge connecting two vertices \(v,w\in\mathcal{M}\suprho\) is denoted by its halfedges \mbox{\(e=\{h_v,h_w\}\)}. Furthermore, we denote by \(V(\mathcal{M}\suprho)\), \(E(\mathcal{M}\suprho)\) and \(E\suprho(\mathcal{M}\suprho)\) the numbers of vertices, ribbon-edges and \(\varrho\)-edges of \(\mathcal{M}\suprho\) and by \(\deg v\) and \(\deg\suprho v\) the number of ribbon- and \(\varrho\)-halfedges at \(v\). The perturbative series writes as a sum over labeled combinatorial maps with \(\varrho\)-edges:
\[
\begin{split}
Z(\lambda) = \sum_{\mathcal{M}\suprho}
\frac{1}{V(\mathcal{M}\suprho)!\, 2^{V(\mathcal{M}\suprho)}} & \Big(\prod_{v\in\mathcal{M}\suprho} \frac{1}{\deg v!\deg\suprho v!}\Big)  \big((-1)^{\abs 1+\abs 2}N_2\big)^{V(\mathcal{M}\suprho)}
\left(-\lambda\right)^{E(\mathcal{M}\suprho)}
\left(-2\kappa\right)^{E\suprho(\mathcal{M}\suprho)}
\crcr 
& \Big( \prod_{v\in\mathcal{M}\suprho} \prod_{\substack{(h_v,h'_v)\\ \text{\tiny corner of}\ v}} g^1_{a_{h_v}b_{h'_v}}
\vphantom{\prod_{e=\{h_v,h_w\}}} \Big)
\Big( \prod_{\substack{e=\{h_v,h_w\} \\ \text{\tiny ribbon-edge}}} 2 P^{b_{h_v}a_{h_v},a_{h_w}b_{h_w}}
\vphantom{\prod_{e=\{h_v,h_w\}}} \Big) \; .
\end{split}
\]

We expand the two terms in each edge projector to sum over ribbon graphs with (twisted) edges and $ \varrho$-edge. This is because the 
amplitude depends on the twisting: every face 
(closed strand) of the ribbon graph contributes a factor of \(N_1\) because along a face an even number of \(g^1\)'s concatenate into a trace. However, it might be necessary to transpose several \(g^1\)'s in order to get this trace: we represent these transpositions by reversing the corresponding arrows along the edge strands and the corners of the ribbon graph, see 
Fig.~\ref{fig: matrix-arrows}. Overall we get (we explain the notation below):
\begin{align}\label{eq: intermedexpansion2}
Z(\lambda) & 
=\sum_{\mathcal{M}\suprho} \frac{1}{V(\mathcal{M}\suprho) !\, 2^{V(\mathcal{M}\suprho)}} \Big(\prod_{v\in\mathcal{M}\suprho} \frac{1}{\deg v!\deg\suprho v!}\Big) \big((-1)^{\abs 1 +\abs 2}N_2\big)^{V(\mathcal{M}\suprho)} \left(-\lambda\right)^{E(\mathcal{M}\suprho)} \left(-2\kappa\right)^{E\suprho(\mathcal{M}\suprho)} \crcr
& \qquad \qquad  \sum_{[\mathcal{G}\suprho]\in Orb_{\mathfrak{T}}(\mathcal{M}\suprho)}
\abs{Stab_{\mathfrak{T}}(\mathcal{G}\suprho)}
\left(N_1\right)^{F(\mathcal{G}\suprho)} \big((-1)^{\abs 1}\big)^{\#\text{transpositions} + \#\text{twists}}
\crcr
&= \sum_{\mathcal{M}\suprho}
\frac{\abs{Stab_{\mathfrak{T}}(\smash[t]{\mathcal{M}\suprho)}} }{V(\mathcal{M}\suprho)!\, 2^{V(\mathcal{M}\suprho)}}\Big(\prod_{v\in\mathcal{M}\suprho} \frac{1}{\deg v!\deg\suprho v!}\Big)
\sum_{[\mathcal{G}\suprho]\in Orb_{\mathfrak{T}}(\mathcal{M}\suprho)} \mathscr{A}(\mathcal{G}\suprho) \; ,
\end{align}
with $ \mathscr{A}(\mathcal{G}\suprho)$ the amplitude of the ribbon graph. Some notation has been introduced in this equation.
Because every ribbon-edge can be twisted or not there are naively \(2^E\) ribbon graphs, associated to \(\mathcal{M}\suprho\). But ribbon graphs are in fact equivalence classes, emphasized by \([\mathcal{G}\suprho]\). Two graphs are equivalent if one can be obtained from the other by successively reversing the order of halfedges of a subset of its vertices and---for each vertex separately---twisting all ribbon-edges connected to these vertices (edges with two twists are again untwisted). As proven in Appendix~\ref{ap: symmfac}, this degeneracy is counted by the cardinal of the stabilizer \(\abs{Stab_{\mathfrak{T}}(\smash[t]{\mathcal{G}\suprho})}\) of the action of a finite group\footnote{\(\mathfrak{T}\) is a subgroup of the so called ribbon group, introduced in \cite{twistedduality}, which also includes the operation of taking the partial dual of a ribbon graph with respect to a subset of its edges.} \(\mathfrak{T}\) whose elements twist a subset of the ribbon-edges of \(\mathcal{G}\suprho\) (\(\mathfrak{T}\) acts trivially on the \(\varrho\)-edges).
In the last step leading to (\ref{eq: intermedexpansion2}) we used the fact that, as
\(\mathfrak{T}\) is Abelian, $Stab_{\mathfrak{T}}(\mathcal{G})=Stab_{\mathfrak{T}}(\mathcal{M})$ for any $\mathcal{G}\in Orb_{\mathfrak{T}}(\mathcal{M})$, where $Orb_{\mathfrak{T}}(\mathcal{M})$ is the orbit of the combinatorial map $ \mathcal{M}$ under the action of $ \mathfrak{T}$. 

For example, the amplitude of the ribbon graph in Fig.~\ref{fig: matrix-arrows} is:
\begin{equation*}
	\underbrace{\left((-1)^{\abs 1 +\abs 2} N_2\right)^3}_{\text{vertices}}
	\underbrace{\left(-\lambda\right)^4}_{\text{edges}}
	\underbrace{\left(-2\kappa\right)^2}_{\varrho\text{-edges}}
	\underbrace{((-1)^{\abs 1})}_{\text{twists}}
	\underbrace{\left(N_1\right)}_{\text{faces}}
	\underbrace{\left((-1)^{\abs 1}\right)^5}_{\substack{\text{arrow} \\ \text{reorientations}}}
	= \left(-\lambda\right)^4\left(-2\kappa\right)^2
	\left((-1)^{\abs 1} N_1\right)
	\left((-1)^{\abs 2} N_2\right)^3
\end{equation*}

\paragraph{Amplitudes.}
The amplitude $\mathscr{A}(\mathcal{G}\suprho)$ can be further computed. 

In Proposition~\ref{thm: standardribbon1}, Appendix~\ref{ap: ribbon} we prove that any ribbon graph can be deformed into a connected sum of:
\begin{itemize}
 \item a graph without twisted edges embeddable into a closed orientable surface \(\Sigma_g\) of genus \(g\)
 \item either no, one or two graphs with a single twisted edge, embeddable into the projective plane \(\mathds{R}\mathds{P}^2\). 
\end{itemize}

This is the ribbon graph equivalent of the classification theorem of closed two dimensional surfaces. 
The crucial observation is that one can track the power of \(-1\) in the amplitude 
under these deformations. In Theorem~\ref{thm: ribbonamplitude} in Appendix~\ref{ap: ribbon} we show that 
for a ribbon graph with twists and which requires transpositions in order to coherently orient the faces:
\[
(-1)^{ V(\mathcal{G}) } \left( \prod_{e\in \mathcal{E}(\mathcal{G})}  (-1)^{\tau(e)} \right) 
\left( \prod_{f\in \mathcal{F}
(\mathcal{G} ) }
  (-1)^{ t(f) } \right) 
  =   (-1)^{ F(\mathcal{G}) } \;,
\]
where $\tau(e)=0$ if the edge $e$ is untwisted (straight) and $\tau(e)= 1 $ if the edge is twisted; $t(f) $ is the number of reorientations of arrows required to coherently orient the face $f$.

The graph $\mathcal{G}\suprho$ can be seen as the union of two ribbon graphs:
\begin{itemize}
 \item one ribbon graph has color $2$ and is trivial. It consisting in all the vertices 
 of $\mathcal{G}\suprho$, each bounded by one face and has no edges. The graph has no twists (as it has no edges) and all its faces  are coherently oriented.
 \item the second one is the graph of color $1$. It has twisted edges and some transpositions are need in order to coherently orient its faces.
\end{itemize}
The amplitude of a graph in \eqref{eq: intermedexpansion2} writes then:
\[
	\mathscr{A}(\mathcal{G}\suprho)= \left(-2\kappa\right)^{E\suprho(\mathcal{G}\suprho)} \left(-\lambda\right)^{E(\mathcal{G}\suprho)} \big( (-1)^{\abs 1} N_1 \big)^{F(\mathcal{G\suprho})}
	\big( (-1)^{\abs 2} N_2 \big)^{V(\mathcal{G}\suprho)} \; ,
\]
proving Eq.~\eqref{eq:ampliintro} in Theorem~\ref{thm: main} for $D=2$.

\paragraph{Combinatorial weights.}
The combinatorial weights in \eqref{eq: intermedexpansion2} simplify by gathering the labeled graphs corresponding to the same unlabeled ribbon graph with \(\varrho\)-edges:
\[
	Z(\kappa,\lambda)=\sum_{[\mathcal{G}\suprho]} \mathscr{W}(\mathcal{G}\suprho)\ \mathscr{A}(\mathcal{G}\suprho) \;,
\]
where the (positive) weights \(\mathscr{W}(\mathcal{G}\suprho)\) include all the combinatorial factors coming from partially resuming (\ref{eq: intermedexpansion2}) to a sum over equivalence classes.

\begin{remark}[Dual graphs]\label{cor: cancellations} The amplitudes of a ribbon graph and its dual are related by:
\[
\mathscr{A}(\mathcal{G}^\ast)  = \left(-\lambda\right)^{E(\mathcal{G}^\ast)}
\big( (-1)^{\abs 1} N_1 \big)^{F(\mathcal{G}^\ast)}
\big( (-1)^{\abs 2} N_2 \big)^{V(\mathcal{G}^\ast)} 
= 	\left(-\lambda\right)^{E(\mathcal{G})}
\big( (-1)^{\abs 2} N_2 \big)^{F(\mathcal{G})}
\big( (-1)^{\abs 1} N_1 \big)^{V(\mathcal{G})}  \;.
\]
We will see below that $\mathscr{W}(\mathcal{G}\suprho) = \mathscr{W}(\mathcal{G}\suprho^\ast) $.
In particular for the mixed \(O(N)\otimes Sp(N)\) models we get $\mathscr{A}(\mathcal{G}\suprho^\ast) = (-1)^{V(\mathcal{G}\suprho)+F(\mathcal{G}\suprho)} \mathscr{A}(\mathcal{G}\suprho)$ hence the contributions of a graph and its dual cancel if \mbox{\(V(\mathcal{G}\suprho)+F(\mathcal{G}\suprho)\)} is odd.

A heuristic argument why $\mathscr{W}(\mathcal{G}\suprho) = \mathscr{W}(\mathcal{G}\suprho^\ast) $ goes as follows.
We split the quartic interactions using an intermediate field \(\sigma_1\) with indices of color-1 coupling to \(M\) via
\( \propto M((g^1\sigma_1 g^1) \otimes g^2)M \). But one can choose the intermediate field to have indices of color-2 and  coupling
\( \propto M(g^1\otimes (g^2\sigma_2 g^2) )M\). The vertices now contribute factors of \(N_1\) and the faces \(N_2\). For any graph, contracting the intermediate field $\sigma_1$ and introducing $\sigma_2$ in the orthogonal channel passes to the dual graph.
\end{remark}

As the combinatorics is insensitive to the symmetry, we focus on the \(O(N_1)\otimes O(N_2)\) model. The connected two-point function of this model:
\[
	G_2(\kappa,\lambda) = Z^{-1}(\kappa,\lambda) \int d\mu[M] \;  \tr\big[ M\delta M^T \delta\big] \;,
\]
obeys a Dyson-Schwinger equation (DSE). Using:
\[
	0 =Z^{-1}(\kappa,\lambda)  \int [dM]\ \frac{\del}{\del M^{a_1a_2}} \Bigg( M^{a_1a_2}e^{- \frac{1}{2}\tr\big[ M\delta M^T\delta \big]
		- \frac{\kappa}{4}\left( \tr\big[ M\delta M^T\delta \big]\right)^2
		- \frac{\lambda}{4}\,\tr\big[\left( M\delta M^T\delta \right)^2\big]} \Bigg)  \;,
\]
we conclude that:
\begin{equation}\label{eq: DSE-G2}
	G_2(\kappa,\lambda) = N_1N_2 + \left( 4\kappa\,\del_\kappa +4\lambda\,\del_\lambda \right) \ln Z(\kappa,\lambda) \; .
\end{equation}

The free energy \(\ln Z\) expands in connected graphs. The derivative operator \(2\kappa\,\del_\kappa + 2\lambda\,\del_\lambda\) generates a rooting of the graph, that we get a sum over graphs with a marked \(\varrho\)- or ribbon-halfedge. Rooted graphs are simpler to count. In Proposition~\ref{prom:combitwopoint}, Appendix~\ref{ap: symmfac} we show that the perturbative series of \(G_2\) writes as:
\begin{align}\label{eq: G2}
	G_2(\kappa,\lambda) = \sum_{[\mathcal{G}\suprho]\text{ connected, rooted}} \frac{1}{2^{C(\mathcal{G}\suprho-\mathcal{E}\suprho)-1}}\ \mathscr{A}(\mathcal{G}\suprho) \; ,
\end{align}
where  $\mathcal{G}\suprho-\mathcal{E}\suprho$ is the graph obtained from $ \mathcal{G}\suprho$ by deleting all the \(\varrho\)-edges and \(C(\mathcal{G}\suprho-\mathcal{E}\suprho) \) denotes the number of its connected components. Note that even if \(\mathcal{G}\suprho\) is connected as a ribbon graph with \(\varrho\)-edges, the graph \(\mathcal{G}\suprho - \mathcal{E}\suprho\) may be disconnected.

It is well known that rooting trivializes the symmetry factors in ordinary combinatorial maps. What is non trivial is that it also simplifies the factor
$2^{-V}  \abs{Stab_{\mathfrak{T}}(\smash[t]{\mathcal{M}\suprho)}} 
 $ in \eqref{eq: intermedexpansion2} to
$2^{- ( C(\mathcal{G}\suprho-\mathcal{E}\suprho)-1)} $. The combinatorial weight in \eqref{eq: G2} is manifestly invariant under duality. Rooted ribbon graphs can be embedded into two dimensional surfaces with one boundary component corresponding to the rooted face.

The Dyson-Schwinger equation for the connected two-point function can be integrated in the sense of formal power series to yield the perturbative expansion of the free energy:
\[
	\ln Z(\kappa,\lambda) =  \sum_{\substack{[\mathcal{G}\suprho]\text{ connected, rooted,}\\ E\text{ or } E\suprho > 0}} \frac{1}{2^{C(\mathcal{G}\suprho-\mathcal{E}\suprho)+1} (E+E\suprho)}\ \mathscr{A}(\mathcal{G}\suprho)  \;,
\]
where $E$ and $ E\suprho$ denote the number of ribbon edges respectively $\varrho$-edges in 
$\mathcal{G}\suprho$. The integration does not spoil the symmetry under duality because the powers of the coupling constants in the amplitude only depend on the numbers of edges. Finally, the partition function \(Z\) can then be obtained by exponentiating \(\ln Z\).

\section{Tensor Models}\label{sec: tensor}

The case $D\ge 3$ is treated similarly to the case $D=2$. However, as the number of available quartic invariants grows exponentially with $D$ (recall Lemma~\ref{lem:invcount}), the number of intermediate fields grows also. Moreover, the intermediate fields are matrices with different dimensions. At most one of the $N_c$ factors can be rendered explicit as a parameter in the integral, and one must rely on graphical methods to track the other $N_c$'s.

In $D\ge 3$ the perturbative expansion is an expansion in \emph{colored multi-ribbon graphs} which can be understood intuitively as stacked ribbon graphs. The \(N_c\) to \(-N_c\) duality holds graph by graph because only the combination \((-1)^{\abs c}N_c\) appears in the amplitude of a graph. 

If one aims to study tensor (or matrix) models with a sensible large $N$ limit one needs to rescale the coupling constants with powers of \(N\). Care has to be taken if one wants to preserve the manifest \(N\) to \(-N\) duality: this can sometimes require a flip of the sign of some of the coupling constants.

\subsection{Intermediate Field Representation}\label{sec: IFRtensor}

Complex random tensor models in the intermediate field representation were, for example, studied in \cite{Gurau:2013pca,quartic}. We introduce an intermediate field per quartic interaction. 
For $\mathcal{C} $ a subset of the colors we denote $\Sigma^{\mathcal{C}}$ the set of \(N^{\abs{\mathcal{C}}}\times N^{\abs{\mathcal{C}}}\) matrices (where $\abs{\mathcal{C}}$ denotes the cardinal of $\mathcal{C}$) taken to be symmetric if the sum of the parities of the indices in $\mathcal{C}$ is even and anti-symmetric if it is odd:
\[
\begin{split}
& \Sigma^{\mathcal{C}} =  
\begin{cases}
Sym^2\left(\bigotimes_{c\in\mathcal{C}}\mathscr{H}_c\right), \quad &\sum_{c\in\mathcal{C}}\abs c = 0 \mod 2 \\
\Lambda^2\left(\bigotimes_{c\in\mathcal{C}}\mathscr{H}_c\right), \quad &\sum_{c\in\mathcal{C}}\abs c = 1 \mod 2
\end{cases}  \;, \crcr
& \sigma^{a_{\mathcal{C}}^1 a_{\mathcal{C}}^2} = (-1)^{\sum_{c\in\mathcal{C}}\abs c} \; \sigma^{a_{\mathcal{C}}^2 a_{\mathcal{C}}^1} \; , \qquad \forall \sigma \in \Sigma^{\mathcal{C}}  \;,
\end{split}
\]
where we recall that $a_{ \mathcal{C}} $ denotes a multi index $(a_c | c\in \mathcal{C} )$.
Note that \(\sigma\) is always commuting (bosonic) because \(\bigotimes_{c\in\mathcal{C}}\mathscr{H}_c\) is either purely odd or even. For \(\mathcal{C}=\emptyset\), set \(\sigma\in\Lambda_\infty^0\) the commuting scalars. 
Since \(\sigma\) are (anti-)symmetric under exchange of their two multi-indices, it is useful to introduce the (anti-)symmetric projector:
\begin{equation}\label{eq: projector}
P_{\mathcal{C}} : \Sigma^{\mathcal{C}} \to \Sigma^{\mathcal{C}} \;,\qquad 
\tensor{(P_{\mathcal{C}} )}{^{a^1_{\mathcal{C}} a^2_{\mathcal{C}}}_{,b^1_{\mathcal{C}} b^2_{\mathcal{C}}} }
:= \frac{1}{2} \Big(
\prod_{c\in\mathcal{C}} \tensor{\delta}{^{a^1_c}_{b^1_c}} \tensor{\delta}{^{a^2_c}_{ b^2_c} }  + (-1)^{\sum_{c\in\mathcal{C}}\abs c} \prod_{c\in\mathcal{C}} \tensor{\delta}{^{a^1_c}_{b^2_c}} \tensor{\delta}{^{a^2_c}_{ b^1_c} }
\Big) \; , 
\end{equation}
and $ P_{\mathcal{C}} $  is the identity for 
\(\mathcal{C}=\emptyset\). The projector is such that:
\[
\tensor{ (P_{\mathcal{C}} )}{^{a^1_{\mathcal{C}} a^2_{\mathcal{C}}}_{, b^1_{\mathcal{C}} b^2_{\mathcal{C}}} } 
= (-1)^{\sum_{c\in\mathcal{C}}\abs c} \tensor{(P_{\mathcal{C}} )}{^{a^2_{\mathcal{C}} a^1_{\mathcal{C}}}_{,b^1_{\mathcal{C}} b^2_{\mathcal{C}}} } = (-1)^{\sum_{c\in\mathcal{C}}\abs c} \tensor{(P_{\mathcal{C}} )}{^{a^1_{\mathcal{C}} a^2_{\mathcal{C}}}_{,b^2_{\mathcal{C}} b^1_{\mathcal{C}}} } \;; \qquad
\smash[t]{\frac{\partial \sigma^{a^1_{\mathcal{C}} a^2_{\mathcal{C}}}}{\partial \sigma^{ b^1_{\mathcal{C}} b^2_{\mathcal{C}} }}}
= \tensor{ (P_{\mathcal{C}} )}{^{a^1_{\mathcal{C}} a^2_{\mathcal{C}}}_{, b^1_{\mathcal{C}} b^2_{\mathcal{C}}} } \; .
\]

\begin{lem}[Hubbard Stratonovich Transformation\label{lem: hubbardstratonovich}]
Every quartic tensor invariant \(I(T)\) in 
Eq.~\eqref{eq: quarticinv}:
\[
\begin{split}
I(T) & = \sum_{a_{\mathcal{D}}^1, a_{\mathcal{D}}^2, b_{\mathcal{D}}^1, b_{\mathcal{D}}^2}
\Big(T^{a_{\mathcal{D}}^1}T^{a_{\mathcal{D}}^2} \prod_{c\in\mathcal{D}\backslash\mathcal{C}} g^c_{a_c^1 a_c^2} \Big)
\Big(T^{b_{\mathcal{D}}^1}T^{b_{\mathcal{D}}^2} \prod_{c\in\mathcal{D}\backslash\mathcal{C}} g^c_{b_c^1 b_c^2} \Big) \Big(\prod_{c\in\mathcal{C}} \big(-\sgn(\pi^c)\big)^{\abs c} g^c_{a_c^1 b_c^{\pi^c(1)}} g^c_{a_c^2 b_c^{\pi^c(2)}}\Big) \\
& = \sum_{a_{\mathcal{C}}^1, a_{\mathcal{C}}^2, b_{\mathcal{C}}^1, b_{\mathcal{C}}^2}
\big(g^{\otimes\mathcal{D}\backslash\mathcal{C}}(T,T)\big)^{a_{\mathcal{C}}^1 a_{\mathcal{C}}^2}
K_{a_{\mathcal{C}}^1 a_{\mathcal{C}}^2, b_{\mathcal{C}}^1 b_{\mathcal{C}}^2}
\big(g^{\otimes\mathcal{D}\backslash\mathcal{C}}(T,T)\big)^{b_{\mathcal{C}}^1 b_{\mathcal{C}}^2} \; ,
\end{split}
\]
with $\pi^c$ fixed permutations of two elements
can (formally) be expressed as a Gau\ss ian integral:
\begin{equation*}
e^{-\frac{\lambda}{4}I(T)} = \Big[e^{\frac{1}{2} (\frac{\del}{\del\sigma} , P_{\mathcal{C}}K\frac{\del}{\del\sigma})} \ e^{-\imath\sqrt{\frac{\lambda}{2}} (\sigma, g^{\otimes\mathcal{D}\backslash\mathcal{C}}(T,T) ) } \Big]_{\sigma=0} ,
\end{equation*}
with 
\(\left(P_{\mathcal{C}}K\right)^{a^1_{\mathcal{C}} a^2_{\mathcal{C}}, b^1_{\mathcal{C}} b^2_{\mathcal{C}}} = \tensor{(P_{\mathcal{C}})}{^{a^1_{\mathcal{C}} a^2_{\mathcal{C}}}_{,c^1_{\mathcal{C}} c^2_{\mathcal{C}}} } K^{ c^1_{\mathcal{C}} c^2_{\mathcal{C}}, b^1_{\mathcal{C}} b^2_{\mathcal{C}}}\) and \(( A,B) = g^{a_1a_2} g^{b_1b_2}A_{a_1b_1}B_{a_2b_2}\), the standard pairing between a vector space and its dual.

\end{lem}
\begin{proof}
The indices of color $c$ of the kernel $K$ are connected as:
\[
 K_{a_{\mathcal{C}}^1 a_{\mathcal{C}}^2, b_{\mathcal{C}}^1 b_{\mathcal{C}}^2}
  \sim \begin{cases}
        (-1)^{|c|} g^c_{a_c^1 b_c^1} g^c_{a_c^2 b_c^2}  \;,\qquad &\pi^c = (1)(2) \\
        g^c_{a_c^1 b_c^2} g^c_{a_c^2 b_c^1}  \;,\qquad & \pi^c = (12) 
       \end{cases} \; ,
\]
hence, as operator, $K^2 = 1$ and $P_{\mathcal{C}}K = KP_{\mathcal{C}} $. Taking into account that 
$g^c_{b_c^1b_c^2} = (-1)^{|c|}g^c_{b_c^2b_c^1}$, and $ \tensor{K}{^{ a^1_{\mathcal{C}} a^2_{\mathcal{C}}}_{, b^1_{\mathcal{C}} b^2_{\mathcal{C}} }  } = \tensor{K}{^{ a^2_{\mathcal{C}} a^1_{\mathcal{C}}}_{, b^2_{\mathcal{C}} b^1_{\mathcal{C}} }  }$, 
and $T^{b_{\mathcal{D}}^1}T^{b_{\mathcal{D}}^2} = (-1)^{\sum_{c\in \mathcal{D}} |c| }T^{b_{\mathcal{D}}^2} T^{b_{\mathcal{D}}^1} $
we have:
\[
 \tensor{K}{^{ a^1_{\mathcal{C}} a^2_{\mathcal{C}}}_{, b^1_{\mathcal{C}} b^2_{\mathcal{C}} }  } 
 \Big(T^{b_{\mathcal{D}}^1}T^{b_{\mathcal{D}}^2} \prod_{c\in\mathcal{D}\backslash\mathcal{C}} g^c_{b_c^1 b_c^2} \Big) =  (-1)^{\sum_{c\in \mathcal{C}} |c| }\tensor{K}{^{ a^2_{\mathcal{C}} a^1_{\mathcal{C}}}_{, b^2_{\mathcal{C}} b^1_{\mathcal{C}} }  } 
 \Big(T^{b_{\mathcal{D}}^2}T^{b_{\mathcal{D}}^1} \prod_{c\in\mathcal{D}\backslash\mathcal{C}} g^c_{b_c^2 b_c^1} \Big) \; ,
\]
that is
$P_{\mathcal{C}}K g^{\otimes\mathcal{D}\backslash\mathcal{C}}(T,T)   = Kg^{\otimes\mathcal{D}\backslash\mathcal{C}}(T,T)$ hence
$ K g^{\otimes\mathcal{D}\backslash\mathcal{C}}(T,T)$ is a matrix with the same symmetry type as $\sigma$. It follows that:
\[
 \big( g^{\otimes\mathcal{D}\setminus\mathcal{C}}(T,T) ,
( P_{\mathcal{C} }K P_{\mathcal{C} }) g^{\otimes\mathcal{D}\setminus\mathcal{C}}(T,T) \big)  = \big( g^{\otimes\mathcal{D}\setminus\mathcal{C}}(T,T), K g^{\otimes\mathcal{D}\setminus\mathcal{C}}(T,T) \big) \;,
\]
hence expanding the exponentials and commuting the sum and the derivative operator we get:
\begin{multline*}
\Big[ \sum_{n=0}^{\infty} \frac{(-\frac{ \lambda}{4})^n} {n!(2n)!} \;
\Big(\frac{\del}{\del\sigma },P_{\mathcal{C} }K \frac{\del}{\del\sigma}\Big)^n
\big(\sigma , g^{\otimes\mathcal{D}\backslash\mathcal{C} }(T,T) \big)^{2n}
\Big]_{\sigma=0}
= \Big[ \sum_{n=0}^{\infty} \frac{
(- \frac{\lambda}{4})^n  }{n!}  \crcr
 \Big(g^{\otimes\mathcal{D}\setminus\mathcal{C}}(T,T)P_{\mathcal{C} })_{a^1_{\mathcal{C}} a^2_{\mathcal{C}}} 
\left(P_{\mathcal{C} }K\right)^{a^1_{\mathcal{C}} a^2_{\mathcal{C}}, b^1_{\mathcal{C}} b^2_{\mathcal{C}}} (P_{\mathcal{C}}g^{\otimes\mathcal{D}\setminus\mathcal{C}}(T,T))_{b^1_{\mathcal{C}} b^2_{\mathcal{C}}}\Big)^{n} \Big]_{\sigma=0} 
= \sum_{n=0}^{\infty} \frac{1}{n!} \Big(- \frac{\lambda}{4} I(T)\Big)^n  \;.
	\end{multline*}
\end{proof}

When dealing with several quartc invariants we will label them $q$ and the corresponding subset of colors \(\mathcal{C}_q\). In order to simplify the notation we sometimes drop this subscript. Using the intermediate fields the partition function of the graded quadratic tensor model of Def.~\ref{def: tensormodel} becomes:
\[
Z(\lambda) = \int \mu [T] = \int [dT] \; e^{
-\frac{1}{2} g^{\otimes\mathcal{D}}(T,T)} \Big[ e^{\sum_{q\in\mathscr{Q}} \frac{1}{2}  (\frac{\del}{\del\sigma_q},P_{\mathcal{C}_q}K_q\frac{\del}{\del\sigma_q})}
	\cdot e^{ - \sum_{q} \imath \sqrt{\frac{\lambda_q}{2}} (\sigma_q,  g^{\otimes\mathcal{D}\setminus\mathcal{C}_q}(T,T) )}
	\Big]_{\sigma_q=0} \; ,
\]
where we denoted the coupling constants generically by $\lambda$.
We denote \(1^{\otimes\mathcal{C}}\) the identity operator acting on  \(\bigotimes_{c\in\mathcal{C}}\mathscr{H}_c \). We define the operator acting on 
$\bigotimes\nolimits_{c=1}^D \mathscr{H}_c$:
\[
A(\lambda) = \sum_{q\in\mathscr{Q}} \imath \sqrt{2\lambda_q}\ 1^{\otimes\mathcal{D}\setminus\mathcal{C}_q}\otimes \sigma_q 
 \; ,\quad
\tensor{A}{^{a^1_{\mathcal{D}}}_{a^2_{\mathcal{D}}}}
= \sum_{q\in\mathscr{Q}} \imath \sqrt{2\lambda_q}\ \Big(\prod_{c\in\mathcal{D}\setminus\mathcal{C}_q} \tensor{\delta}{^{a^1_c}_{a^2_c} } \Big) \;  
\tensor{(\sigma_q)}{^{ a^2_{\mathcal{C}_q } }_{a^1_{\mathcal{C}_q } } }   \;,
\]
and perform the gau\ss ian integral over \(T\) to obtain the partition function in the \emph{intermediate field representation}:
\begin{equation}\label{eq: tensorintermedrep}
	Z(\lambda) =\Big[ e^{ \sum_{q\in\mathscr{Q}} \frac{1}{2}(\frac{\del}{\del\sigma_q},P_{\mathcal{C}_q}K_q\frac{\del}{\del\sigma_q})}\
	e^{	- \frac{(-1)^{\sum_{c\in \mathcal{D}}\abs c}}{2} \tr\ln\left(1^{\otimes\mathcal{D}}+ A(\lambda) \right)}
    \Big]_{\sigma_q=0}
	.
\end{equation}

This is the generalization of Eq.~\eqref{eq: intermedrep2} to $D>3$. The \emph{resolvent operator} for tensors is 
$R = (1^{\otimes\mathcal{D}}+ A(\lambda) )^{-1}$.
The field \(\varrho\) we encountered in $D=2$ corresponds to the unique disconnected quartic invariant \(\mathcal{C}_q=\emptyset\). For now we keep all factors \(N_c\) in the trace: the trace over the color-1 space can be performed explicitly because \(1\notin\mathcal{C}_q \) for all  \(q\in\mathscr{Q}\). In strict generalization of the matrix cases, the sign 
\((-1)^{\sum_{c\in \mathcal{D}}\abs c}\) accounts for 
fermionic/bosonic nature of the tensor field $T$.  Each intermediate field $\sigma_q$ has its own symmetry captured by the sign \((-1)^{\sum_{c\in\mathcal{C}_q}\abs c}\).
The effect of the Hubbard Stratonovich transformation on the Feynman diagrams is depicted Schematically in Fig.~\ref{fig: hubbardstratonovich}.

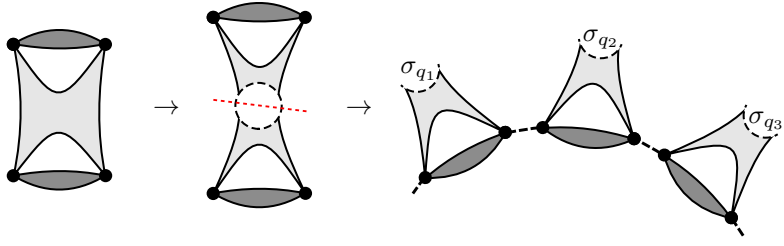
\begin{figure}[ht]
	\centering
	\begin{tikzpicture}[scale=0.875,font=\small]
		\node {\includegraphics[scale=3.5]{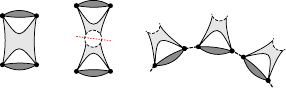}};
		\node at (-3.4,0.15) {\(\rightarrow\)};
		\node at (-.5,0.15) {\(\rightarrow\)};
		\node at (0.4,.73) {\(\sigma_{q_1}\)};
		\node at (3.14,1.22) {\(\sigma_{q_2}\)};
		\node at (5.63,-.02) {\(\sigma_{q_3}\)};
	\end{tikzpicture}
	\caption{The Hubbard-Stratonovich transformation.}
	\label{fig: hubbardstratonovich}
\end{figure}

\subsection{Perturbative Expansion}\label{sec: tensor-if}

\begin{figure}[ht]
	\centering
	\begin{tikzpicture}[scale=0.9,font=\scriptsize]
		\node {\includegraphics[scale=4.5]{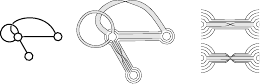}};
		\node at (-5,1.5) {\(\{2\}\)};
		\node at (-4.7,-.05) {\(\{2,3\}\)};
		\node at (-5.95,-.5) {\(\{1,2,3\}\)};
		\node at (-6.9,.5) {\(\{1\}\)};
	\end{tikzpicture}
	\caption{\emph{Left:} Edge multicolored combinatorial map with \(D=3\). Each edge carries a subset of colors. \emph{Center:} A multi-ribbon graph obtained from this multicolored combinatorial map. \emph{Right:} Multi-ribbon edge corresponding to the quartic invariant of Fig.~\ref{fig: quartic inv} in its untwisted (\emph{top}) and twisted (\emph{down}) state.}
	\label{fig: multiribbon}
\end{figure}
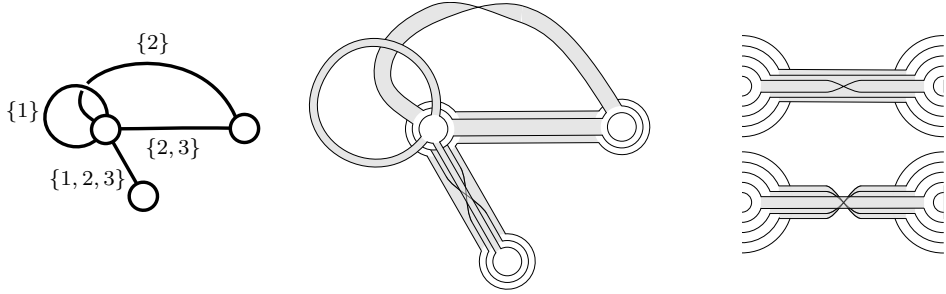

Because of the tensor products, the Feynman graphs of the perturbative expansion of \eqref{eq: tensorintermedrep} are \emph{\(D\)-colored multi-ribbon graphs}. Intuitively they can be understood as \(D\) stacked ribbon graphs. 
Ribbon graphs are obtained from combinatorial maps by replacing their edges by ribbon edges
which can then be twisted or not. Similarly,
\(D\)-colored multi-ribbon graph are obtained from
\emph{edge multicolored combinatorial maps}. These, in turn, are combinatorial maps with edges labeled by subsets of colors \(\mathcal{C}\subset\mathcal{D}\).

\begin{definition}[Edge multicolored Combinatorial Map \cite{gurau}]\label{def: c-map-color}
An \emph{edge multicolored combinatorial map} \(\mathcal{M}\), depicted in Fig.~\ref{fig: multiribbon} on the left,
is composed of:
\begin{itemize}
\item a finite set \(\mathcal{S}\) that is the disjoint union of sets \(\mathcal{S}^{\mathcal{C}}\) of halfedges of the colors \(\mathcal{C}\in\mathcal{D}\), all of even cardinality
$ \mathcal{S} = \bigsqcup_{\mathcal{C}\subset\mathcal{D}} \mathcal{S}^{\mathcal{C}} $.

\item a permutation \(\pi\) on \(\mathcal{S}\).

\item for every \(\mathcal{C}\in\mathcal{D}\) an involution \(\alpha^{\mathcal{C}}\) on \(\mathcal{S}^{\mathcal{C}}\) with no fixed points. The involution \(\alpha^{\mathcal{C}}\) can be extended to the whole of \(\mathcal{S}\) by setting \(\alpha^{\mathcal{C}}(h)=h\ \forall h\in\mathcal{S}\backslash\mathcal{S}^{\mathcal{C}}\).
\end{itemize}

The set of cycles of \(\pi\) is the set of vertices of the map \(\mathcal{V}(\mathcal{M})\). The set of cycles of \(\alpha^{\mathcal{C}}\) is the set of edges of colors \(\mathcal{C}\), \(\mathcal{E}_{\mathcal{C}}(\mathcal{M})\), and \(\mathcal{E}(\mathcal{M})=\bigcup_{\mathcal{C}\subset\mathcal{D}} \mathcal{E}_{\mathcal{C}}(\mathcal{M})\) is the set of all the edges of the map. The cardinalities of these sets are denoted by \(V(\mathcal{M}), E_{\mathcal{C}}(\mathcal{M}),E(\mathcal{M})\) respectively.

An edge multicolored combinatorial map is connected iff the group freely generated by \(\pi\) and the \(\alpha^{\mathcal{C}}\) acts transitively on \(\mathcal{S}\).
\end{definition}

The following definition of multi-ribbon graphs is a generalization of signed rotation systems \cite{embeddedgraphs} which are equivalent to ribbon graphs.

\begin{definition}[\({D}\)-colored Multi-Ribbon Graph]\label{def: multiribbon}
A \emph{\(D\)-colored multi-ribbon graph} 
\(\mathds{G}\), depicted in Fig.~\ref{fig: multiribbon} in the center,
is an edge multicolored combinatorial map \(\mathcal{M}\) equipped with $|\mathcal{C}|$  binary variables taking values $0$ or $1$ on each edge with colors $\mathcal{C}$ (for each edge we have either a $0$ or a $1$ for each of its colors):
\[
\text{for}\ e\in\mathcal{E}_{\mathcal{C}}(\mathcal{M}) \;, \qquad 
e\mapsto\tau(e) = \Big\{ \tau^c(e) \in\{0,1\}|  c\in \mathcal{C} \Big\} \;.
\]

These edges are called (twisted) multi-ribbon edges. Twisting a multi-ribbon edge \(e\) amounts to flipping all the variables \(\tau(e)\), that is $\tau^c(e) \to  \tau^c(e) + 1 \mod 2$.

Two \(D\)-colored multi-ribbon graphs are equivalent if they differ by reversing the order of halfedges around a vertex and simultaneously twisting every incident multi-ribbon edge (self-loops are twisted twice) at a finite number of vertices. 
\end{definition}

The following graphical representation is depicted in  Fig.~\ref{fig: multiribbon}. The vertices of a multi-ribbon graph are represented by \(D\) concentric discs with colors ordered from the innermost to the outermost circle. A multi-ribbon edge \(e\in\mathcal{E}_{\mathcal{C}}(\mathcal{M})\) connects the discs with colors in \(\mathcal{C}\) of its end vertices by ribbon edges. Only discs of the same color can be connected and the ribbons carry the color of the discs they are connecting. 
A \(0\)/\(1\) value of \(\tau^c(e)\) indicates that the ribbon with the color \(c\) of the edge $e$ is un-/twisted. The whole multi-ribbon edge is called untwisted if the ribbon of biggest color in \(\mathcal{C}\) is untwisted. 
The \(\varrho\)-edges encountered in Section~\ref{sec: matrix} are the edges with colors \(\mathcal{C}=\emptyset\). They can be represented as dashed.

The \emph{faces} of color \(c\) of $\mathds{G}$ are the closed circuits obtained by going along the sides of the ribbon edges and along the disks of the vertices of color $c$. The set of faces of color $c$ of \(\mathds{G}\) is denoted $\mathcal{F}_c(\mathds{G})$ and it cardinal is denoted \(F_c(\mathds{G})\). The restriction of \(\mathds{G}\) to a single color \(\mathds{G}_c\) is obtained by deleting all the disks and ribbon with other colors. 
\(\mathds{G}_c\) is an ordinary ribbon graph, possibly disconnected. Observe that \(F_c(\mathds{G})\) is also the number of faces of the ribbon graph \(\mathds{G}_c\).

The perturbative expansion of \eqref{eq: tensorintermedrep} is obtained by Taylor expanding and commuting the sum and the gau\ss ian integral:
\[
Z(\lambda )=
\sum_{V=1}^{\infty} \frac{(-1)^{V \sum_{c=1}^D\abs c}}{V! 2^V}
\Big[
e^{\sum_{q\in\mathscr{Q}}\frac{1}{2}  (\frac{\del}{\del\sigma_q},P_{\mathcal{C}_q}K_q\frac{\del}{\del\sigma_q})}
	\,\Big(-\tr\ln\big(1^{\otimes\mathcal{D}}+ A \big)\Big)^V
	\Big]_{\sigma_q=0}  \; ,
\]
where we suppressed the argument of $A$.
Each \(\tr\ln(1^{\otimes\mathcal{D}}+A)^{-1}\) represents a multi-ribbon vertex. 
The derivatives:
\begin{align*}
 \frac{\del}{\del\sigma_q} \big(1^{\otimes\mathcal{D}}+A\big)^{-1} & = \big(1^{\otimes\mathcal{D}}+A\big)^{-1}
\left( - \frac{\del A}{\del\sigma_q} \right) \big(1^{\otimes\mathcal{D}}+A\big)^{-1} \;, \crcr
\frac{\del}{\del\sigma_q} \tr\ln\big(1^{\otimes\mathcal{D}}+A\big)^{-1} & = \tr\big[\big(1^{\otimes\mathcal{D}}+A\big)^{-1} \left( - \frac{\del A}{\del\sigma_q} \right) \big] \; , \crcr
\frac{\del A}{\del\sigma_q} &=
\imath\sqrt{2\lambda_q} 1^{\otimes\mathcal{D}\backslash\mathcal{C}_q} \otimes P_{\mathcal{C}_q}  \;,
\end{align*}
create multi-ribbon halfedges which, because of the projector, are joined in a twisted or untwisted way. The possible types of multi-ribbon edges depend on the quartic invariants \(q\in\mathscr{Q}\): for brevity the multi-ribbon edges associated to the quartic invariant \(q\) are called \(q\)-edges. The trace induces a cyclic ordering around the vertex which by convention we take to be counter-clockwise. 
Following an index of color \(c\), it goes around the vertex until it encounters a multi-ribbon halfedge with \(c\in\mathcal{C}_q\). As in the matrix case, the order of indices is important if $|c|=1$. This is accounted for by orienting the strands of a vertex in a counter-clockwise manner (Fig.~\ref{fig: multiribbon-arrow}). Denoting $R = \big(1^{\otimes\mathcal{D}}+A\big)^{-1}$, the contribution of an edge writes:
\[
 (-\imath \sqrt{2\lambda_q}) \;  R_{\cdot b^1}  R_{a^1\cdot}  (P_{\mathcal{C}_q } K_q)^{ b^1a^1  , b^2a^2} R_{\cdot b^2} R_{a^2 \cdot } \; ,
\]
and upon setting $\sigma_q=0$ all the resolvents reduce to the identity operator.

\begin{figure}[ht]
	\centering
	\begin{tikzpicture}[scale=0.9,font=\scriptsize]
		\node {\includegraphics[scale=4.5]{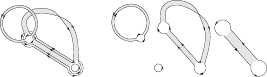}};
		\node at (-2,0) {\(\rightarrow\)};
	\end{tikzpicture}
	\caption{\emph{Left:} A 3-colored multi-ribbon graph. The arrows indicate the order of indies of the \(g^c_{a_c b_c}\). In the a priori orientation arrows point counter-clockwise around vertices and parallel along edges. \emph{Right:} Ribbon graphs obtained by restricting to a single color. The black arrows had to be flipped to arrive at a coherent orientation along each face. Compare to Fig.~\ref{fig: matrix-arrows}.}
	\label{fig: multiribbon-arrow}
\end{figure}

We denote ${\cal M}$ the edge multicolored maps and $ \deg^{q}v $ the number of edges of type $q$ incident at the vertex $v$. As in the matrix case, the edges with  $\mathcal{C}=\emptyset$  are special. We call them $\varrho$-edges and we denote sometimes the number of such edges $E\suprho(\mathcal{M})$. However, note that the $\varrho$-edges are also counted as a particular case $q$-edges for $q\in Q$.

The halfedges incident at a vertex have colors $\mathcal{C}_q$ and we denote them $h^{\mathcal{C}_q}_v$, $f^{\mathcal{C}_{q'}}_v$ and so on. Each half edge is composed of $|\mathcal{C}_q|$ ribbon half edge, one for each color in $\mathcal{C}_q$. The corners\footnote{We exclude the $\varrho$ halfedges when identifying the corners.} of the map $\cal M$ are the pieces of vertices comprised between two consecutive halfedges and we denote them $(h^{\mathcal{C}_q}_v,f^{\mathcal{C}_{q'}}_v)$, with $f^{\mathcal{C}_{q'}}_v $ the successor of $h^{\mathcal{C}_q}_v $ when turning around $v$. 
The partition function becomes:
\begin{align*}
Z(\lambda) =& \sum_{\mathcal{M} }
\frac{1}{V(\mathcal{M})!\, 2^{V(\mathcal{M})}}  \Big(\prod_{v\in\mathcal{M}}  \frac{1}{
\prod_{q} \deg^q v!}\Big)  \big((-1)^{\sum_{c\in \mathcal{D} }  \abs c }\big)^{V(\mathcal{M} )} \;
2^{E\suprho(\mathcal{M})}
\prod_{q \in Q } \left(-\lambda_q \right)^{E_q( \mathcal{M} ) }
\\ 
&   \smashoperator[l]{\prod_{v\in \mathcal{V}(\mathcal{M} ) }} \Big( \prod_{\substack{ (h^{\mathcal{C}_q}_v,f^{\mathcal{C}_{q'}}_v) \\ \text{\tiny corner of}\ v}} \prod_{c\in \mathcal{D} }  g^c_{(a_c)_{h^{\mathcal{C}_q}_v}(b_c)_{f^{\mathcal{C}_{q'}}_v}}
\vphantom{\prod_{e=\{h_v,h_w\}}} \Big)
\\
&    \smashoperator[l]{\prod_{ e=\{h^{\mathcal{C}_q}_v,h^{\mathcal{C}_q}_w\} \in \mathcal{E}(\mathcal{M})}}
 \Big(  ( 2 P_{{\mathcal{C}_q}} K_q)^{ (b_{\mathcal{C}_q } )_{h^{\mathcal{C}_q}_v}(a_{\mathcal{C}_q})_{h^{\mathcal{C}_q}_v},(b_{\mathcal{C}_q})_{h^{\mathcal{C}_q}_w} (a_{\mathcal{C}_q})_{h^{\mathcal{C}_q}_w} }
\vphantom{\prod_{e=\{h_v,h_w\}}} 
  \prod_{c\notin C_q} (g^c)^{(b_{c})_{h^{\mathcal{C}_q}_v} (a_c)_{h^{\mathcal{C}_q}_v}}(g^c)^{(b_{c})_{h^{\mathcal{C}_q}_w} (a_c)_{h^{\mathcal{C}_q}_w}}
\Big) \; ,
\end{align*}
with the convention that if $\mathcal{C}_q = \emptyset$, then there is no corner and $ ( 2 P_{{\mathcal{C}_q}}K_q)=1$.

An index of color $c$ is insensitive to the halfedges with colors different from $c$: an index follows a face and closes in a trace when the face closes. As in the matrix case, we obtain either straight edges or twisted ones coming from the two terms in $P_{{\mathcal{C}_q}}$. In turn, the edges contract on $K_p$ kernels that send the color $c$ either in a parallel channel or in a cross one. Overall, the ribbon  of color $c$ of the edge $( 2 P_{{\mathcal{C}_q}}K_q) $ can either be straight, which we denote $\tau^c(e) = 0$ or twisted, denoted $\tau^c(e) = 1$. Let us track the indices of color $c$ coming from a term in $P_{{\mathcal{C}_q}}$ and one possible $K_q$, for instance:
\[
 ( 2 P_{{\mathcal{C}_q}}K_q) \sim 
  \tensor{\delta}{^{b_c^1}_i} \tensor{\delta}{^{a_c^1}_j} (-1)^{|c|} g^{i\kern.05em b_c^2}g^{j\kern.05em a_c^2} 
   = (-1)^{|c|} g^{b_c^1b_c^2} g^{a_c^1a_c^2} \;.
\]
As this term contracts the indices $b$ together and the $a$ together, it corresponds to a ribbon of color $c$ which is twisted. Proceeding similarly for all the edges and recalling that some $g$'s need to be transposed in order to orient coherently the faces we conclude that:
\begin{equation}\label{eq:partit}
\begin{split}
 Z(\lambda) = &\sum_{\mathcal{M} }
\frac{
\abs{Stab_{\mathfrak{T}}(\smash[t]{\mathcal{M})}} 
}{V(\mathcal{M})!\, 2^{V(\mathcal{M})}}  \Big(\prod_{v\in\mathcal{M}}  \frac{1}{
\prod_{q} \deg^q v!}\Big)  
\sum_{[\mathds{G}]\in Orb_{\mathfrak{T}}(\mathcal{M} )} \mathscr{A}(\mathds{G} ) \; , \\
 \mathscr{A}(\mathds{G} ) = &
 \big((-1)^{\sum_{c\in \mathcal{D} }  \abs c }\big)^{V(\mathds{G} )} \;
2^{E\suprho(\mathcal{\mathds{G}})}
\prod_{q \in Q } \left(-\lambda_q \right)^{E_q( \mathds{G} ) } 
 \left( \prod_{C_q\neq \emptyset} \prod_{e \in \mathcal{E}_{\mathcal{C}_q}(\mathds{G})} 
  (-1)^{ \sum_{c\in \mathcal{C}_q} \tau^c(e) \; \abs c}
 \right)  \\
& \hspace{7.8em}  \prod_{c} 
  \prod_{f\in \mathcal{F}_c(\mathds{G})}
  (-1)^{ t(f) \; \abs c } N_c\;,
\end{split}
\end{equation}
where $t(f)$ denotes the number of transpositions needed to orient the face $f$ coherently.

\paragraph{Amplitudes.} Up to the overall coupling constants, the amplitude of a graph factors over the ribbon graphs 
$\mathds{G}_c$:
\[
  \mathscr{A}(\mathds{G} ) = 
  2^{E\suprho(\mathcal{\mathds{G}})}
\prod_{q \in Q } \left(-\lambda_q \right)^{E_q( \mathds{G} ) } \prod_{c \in \mathcal{D} } \left[ 
  (-1)^{ V(\mathds{G}_c) \; |c| } \left( \prod_{e\in \mathcal{E}(\mathds{G}_c)}  (-1)^{\tau^c(e) \; \abs c}\right) 
\left( \prod_{f\in \mathcal{F}_c(\mathds{G})}
  (-1)^{ t(f) \; \abs c } N_c \right) 
\right] \;, 
\] 
and using Theorem~\ref{thm: ribbonamplitude} in Appendix~\ref{ap: ribbon} this is:
\begin{equation}\label{eq: tensor-amplitude}
	\mathscr{A}(\mathds{G}) = 2^{E\suprho(\mathds{G})}\prod_{q\in\mathscr{Q}} \left(-\lambda_q\right)^{E_q(\mathds{G})} \prod_{c=1}^D \big((-1)^{\abs c} N_c\big)^{F_c(\mathds{G})} \; .
\end{equation}
and thus obeys the \(N_c\to -N_c\) duality. The  \(\varrho\)-edges edges associated to the unique disconnected invariant \(\mathcal{C}_q=\emptyset\) do not have a twisted or untwisted state and bring a relative factor of two compared to the other multi-ribbon edges.

\paragraph{The two-point function.}
The connected two-point function of the tensor model:
\begin{equation*}
	G_2(\lambda ):= Z^{-1}(\lambda) \int d\mu[T] \;  T^{a_{\mathcal{D}}}T^{b_{\mathcal{D}}} g^1_{a_1b_1}\dots g^D_{a_Db_D} 
\end{equation*}
can be expressed as a perturbative series over rooted multi-ribbon graphs. As in the matrix case, rooting drastically simplifies the combinatorial factors. The DSE for \(G_2\) follows from:
\begin{equation}\label{eq: DSE-tensor}
0=\int dT \frac{\del}{\del T^{a_{\mathcal{D}}}} \big( T^{a_{\mathcal{D}}} e^{-S[T]} \big) 
\quad\Rightarrow \qquad 
G_2(\lambda) = \prod_{c\in \mathcal{D}} N_c + \sum_{q\in\mathscr{Q}} 4\lambda_q \del_{\lambda_q} \ln Z .
\end{equation}
Graphically, the derivatives select an edge of a multi-ribbon graph and because every edge has two halfedges, \(\sum_{q\in\mathscr{Q}} 2\lambda_q \del_{\lambda_q}\) generates a sum over all possible rootings. Rooted unlabeled multi-ribbon graphs are equivalence classes of labeled multi-ribbon graphs that differ only by relabeling of their halfedges, but keeping the root halfedge fixed.
The calculation of \  \(\abs{Stab_{\mathfrak{T}}(\mathds{G})}\) is a straightforward generalization of the ordinary ribbon graph case and in Proposition~\ref{prom:combitwopoint}
Appendix~\ref{ap: symmfac} we show:
\begin{equation}\label{eq: G2-tensor}
G_2( \lambda) = \sum_{\lbrack\mathds{G}\rbrack\ \text{connected, rooted}}
\frac{1}{2^{C(\mathds{G}-\mathcal{E}\suprho )-1}} \mathscr{A}(\mathds{G}) \; ,
\end{equation}
where \(C(\mathds{G}-\mathcal{E}\suprho)\) counts the number of connected components of the multi-ribbon graph obtained after deletion of the \(\varrho\)-edges. The free energy 
\(\ln Z( \lambda )\) can be obtained by integrating the DSE:
\begin{equation}\label{eq: lnZ-tensor}
	\ln Z(\lambda) = \sum_{\substack{\lbrack\mathds{G}\rbrack\ \text{connected, rooted,} \\ \text{at least one } E_q>0}}
	\frac{1}{2^{C(\mathds{G}-\mathcal{E}\suprho)+1} \sum_{q\in\mathscr{Q}} E_q(\mathds{G})}
	\mathscr{A}(\mathds{G}) \; .
\end{equation}

\paragraph{Rescaled theories.} Models which admit a good $1/N$ expansion involve couplings rescaled by various powers of $N$. In order to maintain the \(N\) to \(-N\) duality of the amplitudes one needs sometimes to flip the sign of the couplings. For instance for $D=2$, in order to get a sensible large $N$ limit one needs to rescale the coupling by a factor $N$.
If one rescales \(\lambda\to \lambda/N\) in the \(O(N)\otimes O(N)\) model and \mbox{\(\lambda\to -\lambda/N\)} in the \(Sp(N)\otimes Sp(N)\) model the amplitudes graphs differs by \((-1)^{\chi(\mathcal{G})}\).\footnote{This was also found in \cite{goe-gse}.} The equality is reestablished if one sends at the same time \(\lambda\to -\lambda\).


\newpage

{\appendix

\section{Classification of  Ribbon Graphs}\label{ap: ribbon}
\subsection{Canonical Form}
	
We prove in this subsection that a ribbon graph can be brought into a canonical form obtained by first separating
the oriented and unoriented parts of the graph  (Proposition~\ref{thm: standardribbon1}) and then simplifying the oriented part (Proposition~\ref{thm: standardribbon2}).

\begin{prop}\label{thm: standardribbon1}
Every connected ribbon graph \(\mathcal{G}\) is homeomorphic as a topological surface (2 dim.\ CW complex) to a ribbon graph \(\mathcal{G}'\) such that:
\begin{itemize}
 \item \(\mathcal{G}'\)  has only one vertex,
 \item \(\mathcal{G}'\)  has either none, or one or two twisted simple self-loops,
 \item all the remaining edges of  \(\mathcal{G}'\) are untwisted.
\end{itemize}
Equivalently:
\[
\mathcal{G} \isom \mathcal{G}' \isom
\begin{cases}
\mathcal{G}_{\Sigma_g}                                  & \text{orientable with}\ k=2g \\
\mathcal{G}_{\Sigma_g} \vee \mathcal{G}_{\mathds{RP}^2} & \text{unorientable with}\ k=2g+1 \\
\mathcal{G}_{\Sigma_g} \vee \mathcal{G}_{\mathds{RP}^2} \vee \mathcal{G}_{\mathds{RP}^2} & \text{unorientable with}\ k=2g+2 
\end{cases} \;,
\]
where \(\mathcal{G}_{\Sigma_g}\) a ribbon subgraph of \(\mathcal{G}'\) containing only untwisted edges and is cellularly embedded into a closed orientable surface  \(\Sigma_g\) with orientable genus 
$g$ (we reserve the notation \(g\) for the orientable genus) and $k$ is the non orientable genus of $\mathcal{G}'$.
\end{prop}

Proposition~\ref{thm: standardribbon1} is illustrated in Fig.~\ref{fig: ap-stdribbon}.

\begin{figure}[ht]
		\centering
		\begin{tikzpicture}[scale=.8333,font=\small]
			\node at (0,0) {\includegraphics[scale=5]{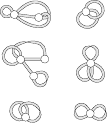}};
			\node at (.6,3) {\(\simeq\)};
			\node at (.6,0) {\(\simeq\)};
			\node at (.6,-3) {\(\simeq\)};
			\node[anchor=west] at (4,3) {\(\isom\mathcal{G}_{T^2}\)};
			\node[anchor=west] at (4,0) {\(\isom\mathcal{G}_{T^2}\vee\mathcal{G}_{\mathds{RP}^2}\)};
			\node[anchor=west] at (4,-3) {\(\isom\mathcal{G}_{\mathds{RP}^2}\vee\mathcal{G}_{\mathds{RP}^2}\)};
		\end{tikzpicture}
		\caption{Illustration of Proposition~\ref{thm: standardribbon1}. In the first line the orientable part was further simplified using Proposition~\ref{thm: standardribbon2}. We call the right hand side the canonical form.}
		\label{fig: ap-stdribbon}
\end{figure}
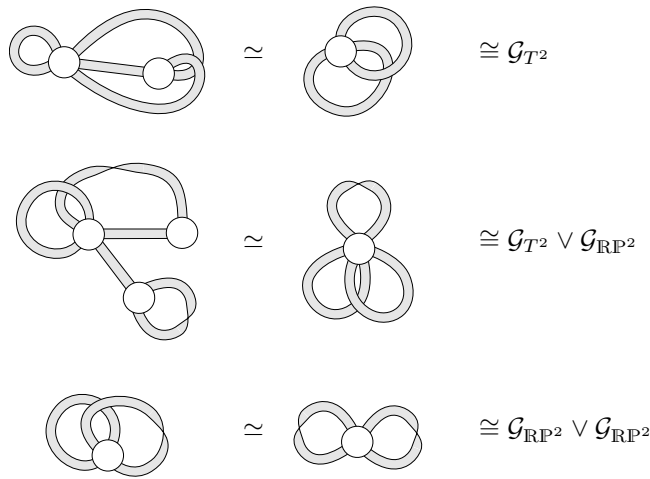

In order to state our second proposition, we need the notion of \emph{clean nice crossing}.

\begin{definition}[Nice Crossing and Clean Nice Crossing] \label{def: nice crossing}
Let \(e=\{e_1,e_2\}\) and \(f=\{f_1,f_2\}\) be two untwisted self-loop edges connected to the same vertex \(v\) of a ribbon graph. Assume \(f_1<f_2\) and \(e_1<e_2\) in the cyclic order around \(v\).
\begin{itemize}
\item The pair \((e,f)\) is a \emph{nice crossing} 	\cite{Gurau2006}, iff \(e_2\) is the successor of \(f_1\).
\item A nice crossing \((e,f)\) is called \emph{clean nice crossing}, if there is no other  halfedge \(h\) of \(v\) distinct from $e_2,f_1$ satisfying \mbox{\(e_1<h<f_2\)}, i.\,e.\ along $v$ the halfedges are encountered in the order $\dots e_1f_1e_2f_2 \dots$.
\end{itemize}
\end{definition}

\begin{prop}\label{thm: standardribbon2}
Every ribbon graph \(\mathcal{G}\) composed of only untwisted edges is homeomorphic as a topological surface (2 dim.\ CW complex) to a ribbon graph \(\mathcal{G}'\) with one vertex, one face and $2g$ edges forming \(g\) clean nice crossings, where \(g\) is the orientable genus of \(\mathcal{G}\). Equivalently:
\[
\mathcal{G}\simeq \mathcal{G}' \isom
\mathcal{G}_{\kern-.1em\mathlarger{\circ}} \,
\bigvee_g \mathcal{G}_{T^2} \quad \text{with}\ \chi=2-2g \; .
\]
\end{prop}

Note that Proposition~\ref{thm: standardribbon2} can be applied to \(\mathcal{G}_{\Sigma_g}\) in Proposition~\ref{thm: standardribbon1}, yielding: 
\begin{equation} \label{eq: canonform}
\mathcal{G} \simeq
\mathcal{G}_{\kern-.1em\mathlarger{\circ}} \,
\bigvee_g \mathcal{G}_{T^2} \bigvee_{0,\, 1\, \text{or}\ 2} \mathcal{G}_{\mathds{RP}^2} 
\; .
\end{equation}
We call the right hand side of this equation the \emph{canonical form} of \(\mathcal{G}\), see Fig.~\ref{fig: ap-stdribbon}.
This is the ribbon graph version of classification theorem of closed surfaces, stating that every such surface is homeomorphic to the connected sum of a sphere, some number of tori, and either no, one or two real projective planes.
	
\paragraph{Contraction and sliding of edges.}
We introduce two homeomorphisms of ribbon graphs, viewed as a topological surface with boundary. Similar moves are known in the literature \cite{grosstucker, handleslides}.

\begin{definition}[Contraction of an Edge, see Fig.~\ref{fig: ap-contraction}]\label{def: contraction}Let \(\mathcal{G}\) be a ribbon graph and \(e\in\mathcal{E}(\mathcal{G})\) an edge connecting two distinct vertices \(v,w\in\mathcal{V}(\mathcal{G})\) of coordination \(p\) and \(q\). Remember that \(v,w\) and \(e\) are all topological disks.

If \(e\) is untwisted we define \(\mathcal{G}/e\) to be the ribbon graph obtained from 
\(\mathcal{G}\) by replacing \(v,w\) and \(e\) with the single vertex \(u=v\cup e\cup w\) (which is again a topological disk) of coordination
\(p+q-2\) such that in the cyclic ordering around this vertex the halfedges of \(v\) proceed the halfedges of \(w\).
The ribbon graph \(\mathcal{G}/e\) has one vertex and one edge fewer than \(\mathcal{G}\), but the same number of faces. 

If \(e\) is twisted we first push the twist along the graph by reembedding the vertex $w$ such that \(e\) is untwisted and proceed as before.

The contraction preserves the Euler characteristic and the orientability and is thus a homeomorphism of surfaces.
\end{definition}

	\begin{figure}[ht]
		\centering
		\begin{tikzpicture}[scale=.8333,font=\scriptsize]
			\node at (0,0) {\includegraphics[scale=5]{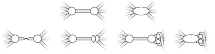}};
			\node at (0.55,.95) {\(\rightarrow\)};
			\node at (-3.17,-.75) {\(\rightarrow\)};
			\node at (0.35,-.75) {\(\rightarrow\)};
			\node at (3.9,-.75) {\(\rightarrow\)};
		\end{tikzpicture}
		\caption{Contraction of an untwisted (\emph{first line}) and twisted (\emph{second line}) edge in a ribbon graph.}
		\label{fig: ap-contraction}
	\end{figure}
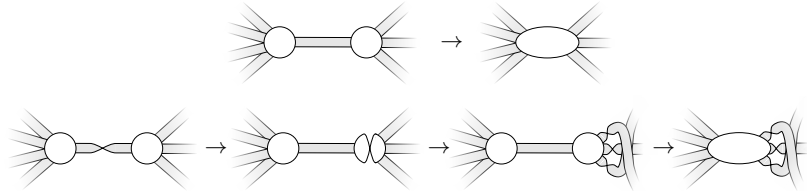
	
A spanning tree of $\mathcal{G}$, that is 
a connected acyclic subgraph \(\mathcal{T}\subset\mathcal{G}\), has 
$E(\mathcal{T}) = V(\mathcal{G})-1 $ edges. 
One can contract all the edges in a spanning tree and decrease the numbers of vertices and edges of $\mathcal{G}$ to $V(\mathcal{G} ) \to  V(\mathcal{G}) - (V(\mathcal{G})-1)  =  1 $ and  \( E(\mathcal{G}) \to
E(\mathcal{G}) - (V(\mathcal{G})-1) \). The resulting graph is a rosette graph homeomorphic to $\mathcal{G}$.

\begin{definition}[Sliding of Edges Ia, see Fig.~\ref{fig: ap-sliding_I}] \label{def: sliding Ia}
Let \(e=\{e_1,e_2\}\) be a twisted self-loop edge on the vertex \(v\) of a ribbon graph. In the cyclic ordering of halfedges around \(v\), let \(e_1<e_2\) and denote by \(e_1< h_1<h_2<\dots <h_n < e_2\) all the halfedges of \(v\) that are between \(e_1\) and \(e_2\).

\emph{Sliding of the halfedges \(h_1,\ldots,h_n\) out of the twisted edge \(e\)} is defined as:
\begin{enumerate}
\item Reordering the halfedges to \(h_n< \dots <h_2<h_1<e_1<e_2\). 
\item Adding a twist (recall that two twists on the same edge cancel) to all the edges to which \(h_1,\ldots,h_n\) belong.
\end{enumerate}

Note that the order of the \(h_i\)'s has been reversed. Also, note that and after the sliding, \(e\) is a simple twisted self-loop.
\end{definition}

\begin{definition}[Sliding of Edges Ib, see Fig.~\ref{fig: ap-sliding_I}] \label{def: sliding Ib}
Let \(e=\{e_1,e_2\}\) be a simple twisted self-loop on the vertex \(v\) of a ribbon graph. In the cyclic ordering of halfedges around \(v\), let \(e<1<e_2< h_1<h_2<\dots <h_n\) a with $h_i$ a collection of consecutive halfedges preceding 
$e_1$ on \(v\). As \(e\) is a simple self-loop, there is no halfedge between \(e_1\) and \(e_2\).

\emph{Sliding of the halfedges \(h_1,\dots,h_n\) past the twisted edge \(e\)} is defined as:
\begin{enumerate}
\item  Reordering the halfedges to \(e_1<e_2<h_1<h_2<\dots <h_n\).
\end{enumerate}

Note that the relative order of the \(h_i\)'s has not changed, no additional twists where introduced and \(e\) remains a simple twisted self-loop.
\end{definition}
	
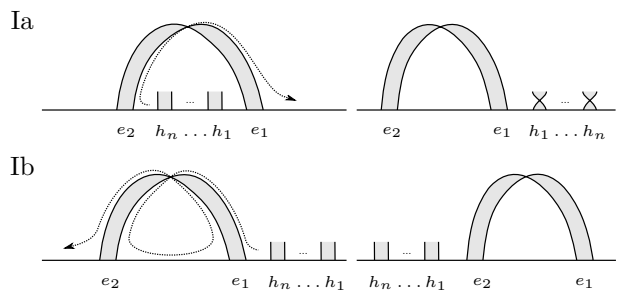
\begin{figure}[ht]
		\centering
		\begin{tikzpicture}[font=\tiny, scale=8/7]
			\node at (0,0) {\includegraphics[scale=8]{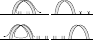}};
			\begin{scope}[shift={(4.27,0)}]
				\node at (-7.8,1.5) {\small{Ia}};
				\node at (-5.06,0.2) {\(e_1\)};
				\node at (-6.60,0.2) {\(e_2\)};
				\node at (-5.82,0.2) {\(h_n\dots h_1\)};
				\node at (-2.27,0.2) {\(e_1\)};
				\node at (-3.54,0.2) {\(e_2\)};
				\node at (-1.52,0.2) {\(h_1\dots h_n\)};
				\node at (-7.8,-0.2) {\small{Ib}};
				\node at (-5.28,-1.55) {\(e_1\)};
				\node at (-6.78,-1.55) {\(e_2\)};
				\node at (-4.52,-1.55) {\(h_n\dots h_1\)};
				\node at (-1.32,-1.55) {\(e_1\)};
				\node at (-2.50,-1.55) {\(e_2\)};
				\node at (-3.36,-1.55) {\(h_n\dots h_1\)};
			\end{scope}
		\end{tikzpicture}
		\caption{Sliding of edges Ia and b. The horizontal line is the vertex with ordering from right to left.}
		\label{fig: ap-sliding_I}
\end{figure}
	
Both sliding operation (Ia) and (Ib) preserve the number of faces, do not change the numbers of vertices and edges and do not alter the orientability. Thus these operations are homeomorphisms of two dimensional surfaces.
	
\begin{definition}[Sliding of Edges IIa, see Fig.~\ref{fig: ap-sliding_II}] \label{def: sliding IIa}
Let \((e=\{e_1,e_2\}, f=\{f_1,f_2\})\) be a nice crossing at the vertex \(v\) of a ribbon graph. In the cyclic ordering of halfedges around \(v\), let us denote
\[
e_1<h_1 \dots < h_n<f_1<e_2<k_1<\dots k_m <f_2 \;,
\]
the halfedges located between $e_1$ and $f_2$. 

\emph{Sliding of the halfedges \(h_1,\ldots,h_n,k_1,\ldots,k_m\) out of the nice crossing \((e,f)\)} is defined as:
\begin{enumerate}
\item Reordering the halfedges to \(k_1< \dots <k_m<h_1<\dots <h_n<e_1<f_1<e_2<f_2\). 
\end{enumerate}
Note that the order of the set of \(h_i\)'s and \(k_j\)'s was interchanged, but the relative order in each set remained unchanged. After sliding, \((e,f)\) is a clean nice crossing. 
\end{definition}

\begin{definition}[Sliding of Edges IIb, see Fig.~\ref{fig: ap-sliding_II}] \label{def: sliding IIb}
Let \((e=\{e_1,e_2\}, f=\{f_1,f_2\})\) be a clean nice crossing at the vertex \(v\) of a ribbon graph. In the cyclic ordering of halfedges around \(v\) let us denote:
\[
h_1<\dots <h_n <e_1<f_1<e_2<f_2 \;,
\]
a collection of consecutive halfedges preceding $e_1$ on \(v\).

\emph{Sliding of the halfedges \(h_1,\ldots,h_n\) past the clean nice crossing \((e,f)\)} is defined as:
\begin{enumerate}
\item Reordering the halfedges to \(e_1<f_1<e_2<f_2<h_1< \dots <h_n\). 
\end{enumerate}
Note that the relative order of the \(h_i\)'s is unchanged; \((e,f)\) remains a clean nice crossing.
\end{definition}

    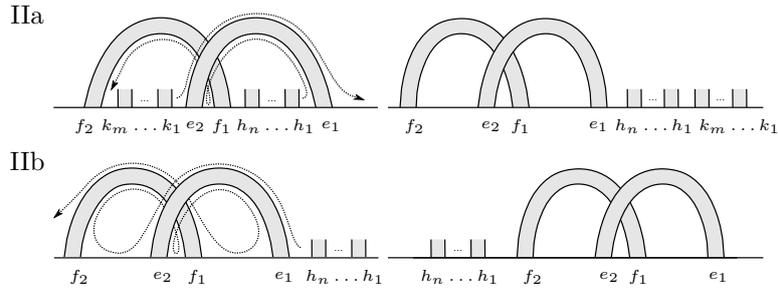
\begin{figure}[ht]
    	\centering
		\begin{tikzpicture}[font=\tiny, scale=8/7]
			\node at (0,0) {\includegraphics[scale=8]{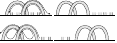}};
			\begin{scope}[shift={(-3.44,0)}]
				\node at (-1,1.5) {\small{IIa}};
				\node at (2.50,0.2) {\(e_1\)};
				\node at (0.94,0.2) {\(e_2\)};
				\node at (1.24,0.2) {\(f_1\)};
				\node at (-0.34,0.2) {\(f_2\)};
				\node at (1.86,0.2) {\(h_n\dots h_1\)};
				\node at (0.32,0.2) {\(k_m\dots k_1\)};
				\node at (5.58,0.2) {\(e_1\)};
				\node at (4.34,0.2) {\(e_2\)};
				\node at (4.68,0.2) {\(f_1\)};
				\node at (3.46,0.2) {\(f_2\)};
				\node at (6.22,0.2) {\(h_n\dots h_1\)};
				\node at (7.2,0.2) {\(k_m\dots k_1\)};
				\node at (-1,-.2) {\small{IIb}};
				\node at (6.96,-1.55) {\(e_1\)};
				\node at (5.70,-1.55) {\(e_2\)};
				\node at (6.04,-1.55) {\(f_1\)};
				\node at (4.82,-1.55) {\(f_2\)};
				\node at (3.98,-1.55) {\(h_n\dots h_1\)};
				\node at (1.96,-1.55) {\(e_1\)};
				\node at (0.56,-1.55) {\(e_2\)};
				\node at (0.96,-1.55) {\(f_1\)};
				\node at (-0.4,-1.55) {\(f_2\)};
				\node at (2.66,-1.55) {\(h_n\dots h_1\)};
			\end{scope}
		\end{tikzpicture}
		\caption{Sliding of edges IIa and b.}
		\label{fig: ap-sliding_II}
	\end{figure}

Like the sliding along twisted edges, the sliding along a nice crossing (IIa, IIb) is a homeomorphisms of two dimensional surfaces.

\begin{proof}[Proof of Proposition~\ref{thm: standardribbon1}] Let \(\mathcal{G}\) be a connected ribbon graph. 

\begin{description}
 \item[First--] contract a spanning tree \(\mathcal{T}\subset\mathcal{G}\). This decreases the number of edges and vertices by \(V(\mathcal{G})-1\) and the resulting ribbon graph \(\mathcal{G}/\mathcal{T}\) is a rosette graph, that is a graph with only one vertex.

\item[Second--] if \(\mathcal{G}/\mathcal{T}\) does not contain any twisted edges then it can be embedded into an orientable surface \(\Sigma_g\) of 
genus \(g\). 

Otherwise, use sliding out of twisted self-loop edges (Ia) to create simple twisted self-loops. This operation may create new twists in the halfedges. Once a twisted self-loop is created, use the slide (Ib) to move it ``to the right'' on the vertex.

Proceed until all the twisted edges of the rosette graph belong to simple twisted self-loops. The resulting graph is a connected sum of an orientable graph \(\mathcal{O}\) containing only untwisted edges and \(p\) copies of \(\mathcal{G}_{\mathds{RP}^2}\), i.\,e.\ ribbon graphs with only one simple twisted loop:
\[
\mathcal{O}\vee \underbrace{\mathcal{G}_{\mathds{RP}^2} \vee \dots \vee \mathcal{G}_{\mathds{RP}^2}}_{p \text{-times}}  
\; .\]

\item[Third--] by sliding as depicted in Fig.~\ref{fig: ap-threetwists}, three neighboring simple twisted self-loops can be reduced to one simple twisted self-loop and a clean nice crossing:
\[
\mathcal{G}_{\mathds{RP}^2}\vee \mathcal{G}_{\mathds{RP}^2} \vee \mathcal{G}_{\mathds{RP}^2} \isom \mathcal{G}_{\mathds{RP}^2} \vee \mathcal{G}_{T^2} 
\;,
\]
hence it is possible to reduce the number of simple twisted self-loops (and twisted edges in total) to zero, one or two. Slide (Ib) the clean nice crossings to the left of the twisted self-loops.

	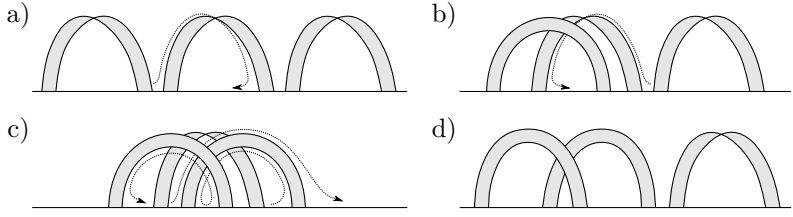
\begin{figure} [ht]
		\centering
		\begin{tikzpicture}[scale=7/8,font=\small]
			\node at (0,0) {\includegraphics[scale=7]{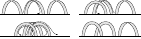}};
			\node at (-5.9,1.55) {a)};
			\node at (0.5,1.55) {b)};
			\node at (-5.9,-0.2) {c)};
			\node at (0.5,-0.2) {d)};
		\end{tikzpicture}
		\caption{Deforming three neighboring simple twisted self-loops into a graph with only one twisted edge. a) and b): By inverting (Ia), slide a halfedge of the left and right twisted simple loop into the central one. This creates a nice crossing. c): Use (IIa) to slide the central twisted loop out of the nice crossing. d): The result has only a single simple twisted loop.}
		\label{fig: ap-threetwists}
	\end{figure}

\item[Finally--] one arrives at a graph \(\mathcal{G}' \) of the form:
\[
\mathcal{G}'\isom 
\begin{cases}
\mathcal{G}_{\Sigma_g}                                  & \text{orientable with}\ k=2g, \\
\mathcal{G}_{\Sigma_g} \vee \mathcal{G}_{\mathds{RP}^2} & \text{unorientable with}\ k=2g+1, \\
\mathcal{G}_{\Sigma_g} \vee \mathcal{G}_{\mathds{RP}^2} \vee \mathcal{G}_{\mathds{RP}^2} & \text{unorientable with}\ k=2g+2,
\end{cases} \;,
\]
with \(\chi(\mathcal{G})=\chi(\mathcal{G}') = 2-k\).
\\ \qedhere
\end{description}
\end{proof}

\begin{proof}[Proof of Proposition~\ref{thm: standardribbon2}] Let \(\mathcal{G}\) be a connected ribbon graph with only untwisted edges. Such a graph can be embedded into an orientable surface.

\begin{description}
 \item[First--] contract a spanning tree \(\mathcal{T}_1\subset\mathcal{G}\) to arrive at a rosette graph \(\mathcal{G}/\mathcal{T}_1\).

 \item[Second--] contract a spanning tree in the dual graph \(\mathcal{T}_2\subset(\mathcal{G}/\mathcal{T}_1)^\ast\). 
This corresponds to deleting edges in 
$ \mathcal{G}$  in a way that preserves the Euler characteristic, the orientability and the connectivity.

This reduces the number of faces and edges by \(F(\mathcal{G})-1\) and gives a superrosette graph \(\mathcal{R}\), that is a graph with one vertex, one face and only untwisted edges.
A superrosette always contains at least one nice crossing. 

\item[Third--] choose a nice crossing \((e,f)\) in \(\mathcal{R} \) and slide (IIa) all the halfedges encompassed by the nice crossing out of \((e,f)\). The result has the structure:
\[
\mathcal{R}\isom \mathcal{R}/(e,f)\vee\mathcal{G}_{T^2} \; ,
\]
where \(\mathcal{R}/(e,f)\) is again a superrosette with genus decreased by one. Iterating one arrives at:
\[
\mathcal{R}\isom \mathcal{G}_{\kern-.1em\mathlarger{\circ}}\vee \underbrace{\mathcal{G}_{T^2}\vee\dots\vee\mathcal{G}_{T^2}}_{g \text{-times}} \;.
\]
\end{description}
\end{proof}

\subsection{Sign of a Ribbon Graph}\label{sec: ribbon-sign}

Let \(\mathcal{G}\) be a connected ribbon graph. An \emph{a priori arrow orientation}\footnote{This is the arrow orientation encountered in Section \ref{sec: tensor}.} of a \(\mathcal{G}\) (which has nothing to do with the orientability of the embedding surface) is defined by:
\begin{enumerate}
\item counter-clockwise pointing arrows at the corners of each vertex. 
\item parallel pointing arrows on the strands of each edges.
\end{enumerate}

We denote $\tau(e)=0$ if the edge $e$ is untwisted (straight) and $\tau(e)= 1 $ if the edge $e$ is twisted. Furthermore, we denote $t(f) $ the number of reorientations of arrows required to coherently orient the face $f$ with all the arrows pointing in the same direction along its boundary. The \emph{sign} of \(\mathcal{G}\) is defined as: 
\[
 \sgn(\mathcal{G}) = (-1)^{ V(\mathcal{G}) } \left( \prod_{e\in \mathcal{E}(\mathcal{G})}  (-1)^{\tau(e)}\right) 
\left( \prod_{f\in \mathcal{F}
(\mathcal{G} ) }
  (-1)^{ t(f) } \right)  \; .
\]

This is well defined. In order to determine the sign of $\mathcal{G}$ one needs to determine the number of arrow flips that are necessary to go from an a priori orientation of \(\mathcal{G}\) to an orientation where all arrows point coherently along the faces of \(\mathcal{G}\) (such an orientation will be called \emph{coherent}). 
As every face consists of as many corners as edge strands, the total number of arrows along a face is even and switching between two coherent orientations requires an even number of arrow flips. Also, as any two a priori orientations differ by an even number of arrow flips (pairs of arrows along the edge strands), switching between a priori orientations at fixed coherent orientation does not change the sign of the graph.

\begin{lem}\label{lem: contract}
 The sign of a graph is: 
 \begin{itemize}
  \item  invariant under reembedding of the vertices.
  \item invariant under contraction of a tree edge.
 \end{itemize}
\end{lem}

\begin{proof}
Consider an a priori arrow orientation of $\mathcal{G}$. Re embedding a vertex of degree $d$ brings 
$d$ new twists, but one needs to reverse $d$ vertex corners in order to orient the re-embedded vertex counterclockwise.

Consider now a tree edge $e$ connecting two vertices $v$ and $w$ in a graph with a priori orientation (which by the first item we can assume to be untwisted). A flip of an arrow coherently orients the disk \(u=v\cup e\cup w\), but this is canceled by the fact that under contraction the number of vertices of the graph goes down by $1$. 

\end{proof}

\begin{lem}\label{lem: slide}
 The sign of a graph is invariant under the sliding moves.
\end{lem}

\begin{proof}

We consider a graph $\mathcal{G}$ having a twisted self-loop as in panel a) Fig.~\ref{fig: ap-sign1}. 

		\begin{figure}[ht]
			\centering
			\begin{tikzpicture}[scale=7/8,font=\tiny]
				\node at (0,0) {\includegraphics[scale=7]{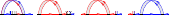}};
				\node at (-7,.9) {\small a)};
				\node at (-5.75,-.7) {\(h_n... h_1\)};
				\node at (-4,.9) {\small b)};
				\node at (-1.2,-.7) {\(h_1... h_n\)};
				\node at (-.4,.9) {\small c)};
				\node at (2.62,-.7) {\(h\)};
				\node at (3.3,.9) {\small d)};
				\node at (4.02,-.7) {\(h\)};
			\end{tikzpicture}
			\caption{Sliding I at a twisted self-loop in a coherently oriented graph. The red and blue corners and strands belong to the red and blue face, respectively. The number of reversed arrows and additional twists is always even.}
			\label{fig: ap-sign1}
		\end{figure}
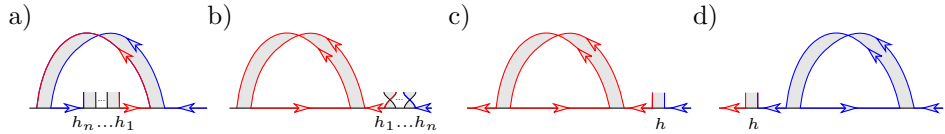

We denote $\mathcal{G}'$ the graph obtained from $\mathcal{G}$ by the sliding Ia. All else being equal, in order to pass from an a priori 
orientation of $\mathcal{G}$ to the coherent orientation depicted in Fig.~\ref{fig: ap-sign1} panel a) two corner arrows had to be reversed, while for $\mathcal{G}'$ only one.
However $\mathcal{G}'$ has one twist more than $\mathcal{G}$. As the graph are otherwise identical they have the same sign.
For Ib sliding there is no extra twist, but both graphs need only one local arrow reorientation.

		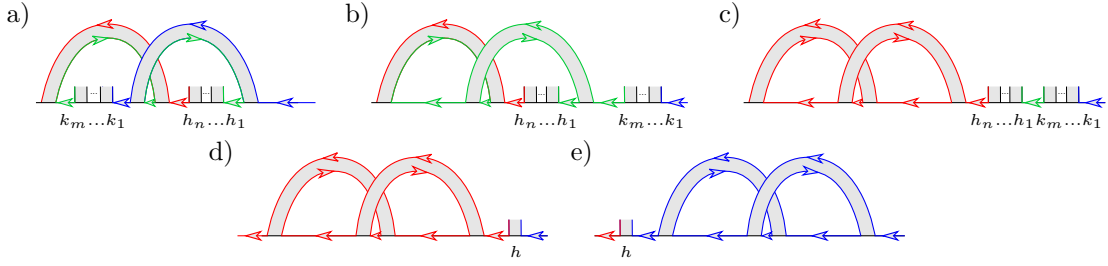
\begin{figure}[ht]
			\centering
			\begin{tikzpicture}[scale=7/8,font=\tiny]
				\node at (0,0) {\includegraphics[scale=7]{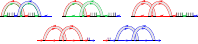}};
				\node at (-8.3,1.8) {\small a)};
				\node at (-7.18,0.2) {\(k_m ...k_1\)};
				\node at (-5.35,0.2) {\(h_n ...h_1\)};
				\node at (-3.2,1.8) {\small b)};
				\node at (1.22,0.2) {\(k_m ...k_1\)};
                \node at (-0.37,0.2) {\(h_n ...h_1\)};
				\node at (2.4,1.8) {\small c)};
				\node at (7.5,0.2) {\(k_m ...k_1\)};
				\node at (6.5,0.2) {\(h_n ...h_1\)};
				\node at (-5.25,-0.3) {\small d)};
				\node at (-0.82,-1.78) {\(h\)};
				\node at (0.17,-0.3) {\small e)};
				\node at (0.84,-1.78) {\(h\)};
			\end{tikzpicture}
			\caption{Sliding II at a nice crossing in a coherently oriented graph. No arrows are reversed, nor are halfedges twisted.}
			\label{fig: ap-sign2}
		\end{figure}

We now consider a graph $\mathcal{G}$ having a nice crossing as in panel a) Fig.~\ref{fig: ap-sign2} and we denote $\mathcal{G}'$ the graph obtained from $\mathcal{G}$ after sliding.
In all the cases, the same number of arrow flips is needed in order to pass locally from an a priori to the coherent orientations depicted. As $\mathcal{G}$ and $\mathcal{G}'$ are identical elsewhere, they have the same sign.

\end{proof}
		
\begin{theorem}[Sign of a Ribbon Graph]\label{thm: ribbonamplitude}
For any connected ribbon graph \(\mathcal{G}\) we have:
\[
 \sgn(\mathcal{G}) = (-1)^{ V(\mathcal{G}) } \left( \prod_{e\in \mathcal{E}(\mathcal{G})}  (-1)^{\tau(e)}\right) 
\left( \prod_{f\in \mathcal{F}
(\mathcal{G} ) }
  (-1)^{ t(f) } \right)  =   (-1)^{ F(\mathcal{G}) } \; ,
\]
with $F(\mathcal{G})$ the number of faces of 
$\mathcal{G}$.
\end{theorem}

\begin{proof}

From Lemmata~\ref{lem: contract} and~\ref{lem: slide}, the sign of a graph is invariant under the reduction moves used in 
Proposition~\ref{thm: standardribbon1}. It follows that, not only:
\[
\mathcal{G} \isom \mathcal{G}' \isom
\begin{cases}
\mathcal{G}_{\Sigma_g}                                  & \text{orientable with}\ k=2g \\
\mathcal{G}_{\Sigma_g} \vee \mathcal{G}_{\mathds{RP}^2} & \text{unorientable with}\ k=2g+1 \\
\mathcal{G}_{\Sigma_g} \vee \mathcal{G}_{\mathds{RP}^2} \vee \mathcal{G}_{\mathds{RP}^2} & \text{unorientable with}\ k=2g+2 
\end{cases} \;,
\]
but also $\sgn(\mathcal{G}) = \sgn (\mathcal{G}')$. The sign of $\mathcal{G}' $ is easy to compute:
\begin{itemize}
 \item $\mathcal{G}'  $ has one vertex
 \item each simple twisted self loop brings a $(-1)$ for the twist and another $(-1)$ in order to change from an a priori arrow orientation  to a coherent one.
 \item the number of untwisted edges of $\mathcal{G}'$ is the number of edges of 
 $\mathcal{G}_{\Sigma_g}$, that is  $E(  \mathcal{G}_{\Sigma_g})$. Exactly one arrow for each such edge needs to be flipped 
in order to pass from an a priori to a  coherent arrow orientation of $\mathcal{G}'  $.
\end{itemize}
Therefore
\[
 \sgn (\mathcal{G}') = (-1)^{1 + E(  \mathcal{G}_{\Sigma_g}) } \; .
\]
The theorem follows by observing that the Euler relation for $ \mathcal{G}_{\Sigma_g}$ reads $ 1 -  E(  \mathcal{G}_{\Sigma_g}) + 
 F( \mathcal{G}_{\Sigma_g})  = 2-2g (\mathcal{G}_{\Sigma_g}) $ hence:
 \[ 
 F( \mathcal{G}_{\Sigma_g}) = 
 1 + E(  \mathcal{G}_{\Sigma_g}) \mod 2 \;,
 \]
and the number of faces is invariant under contraction and sliding $F(\mathcal{G}_{\Sigma_g} ) = F(\mathcal{G}') = F(\mathcal{G}) $.

\end{proof}

\section{Symmetry Factors}\label{ap: symmfac}

The aim of this section is to prove the following Proposition.
\begin{prop}\label{prom:combitwopoint}
The perturbative series of the two point function
\[
G_2((\lambda)):= Z^{-1}(\lambda) \int d\mu[T] \;  T^{a_{\mathcal{D}}}T^{b_{\mathcal{D}}} g^1_{a_1b_1}\dots g^D_{a_Db_D} \;,
\]
write as the sum:
\begin{equation*}
G_2(\kappa,\lambda) = \sum_{[\mathcal{G}\suprho]\ \text{connected, rooted}} \frac{1}{2^{C(\mathcal{G}\suprho-\mathcal{E}\suprho)-1}} \mathscr{A}(\mathcal{G}\suprho) \; .
\end{equation*}
\end{prop}

Before proving this Proposition we discuss some useful facts.
The symmetry factor of a ribbon graph in the perturbative series \eqref{eq: intermedexpansion2} of \(Z(\lambda)\) is obtained as:
\begin{itemize}\normalfont
\item a factor \(\frac{1}{2^V n_{\mathrm{IVP}}}\), where \(n_{\mathrm{IVP}}\) is the number of permutations of vertex labels, that give the same labeled map.
\item a factor \(\frac{1}{\deg v! \deg\suprho v!}\) for every vertex.
\item a factor counting the number of ways to connect labeled halfedges to form the same combinatorial map \(\mathcal{M}\suprho\) underlying \(\mathcal{G}\suprho\), taking into account the different ways to label the halfedges.
\item a factor \(\abs{Stab_{\mathfrak{T}}}\) and the number of combinatorial maps such that \(\mathcal{G}\suprho\) is contained in their orbits under \(\mathfrak{T}\).
\end{itemize} 
For example, the weight of the ribbon graph in Fig.~\ref{fig: matrix-arrows} is:
\begin{equation*}
\underbrace{ \frac{1}{2^3\cdot 1} }_{2^V n_{\mathrm{IVP}}}
\underbrace{ \frac{\overbrace{\ 2^2\ }^{\abs{Stab_{\mathfrak{T}}}} }{ 5!\cdot 1!\cdot 2!\cdot 0!\cdot 1!\cdot 3} }_{\prod_v \deg v!\deg\suprho v!}
\cdot
\overbrace{2\cdot 5!\cdot 3!}^{\substack{\text{labeling and}\\ \text{connecting halfedges}}}
=\frac{1}{2} \; .
\end{equation*}

\paragraph{Stabilizer of Rooted Ribbon Graphs with respect to \(\mathfrak{T}\).}

Rooting simplifies the calculation of \ \(\abs{Stab_{\mathfrak{T}}(\smash[t]{\mathcal{M}})}\). The finite group \(\mathfrak{T}\) that twists the ribbon edges is defined on graphs with a fixed but arbitrary labeling of their edges. The rooting can be used to induce such a labeling: Fix a spanning tree and enumerate all edges as they are encountered on a counter-clockwise walk following the unique face of the tree, starting at the root.

We first focus on ordinary combinatorial maps and ribbon graphs. The results can be generalized to graphs with \(\varrho\)-edges, by considering the ordinary ribbon graph that is obtained by deleting the \(\varrho\)-edges.

\begin{lem}\label{thm: stab}
Let \(\mathcal{G}\) be a rooted, connected ribbon graph. We denote by \(V_1\) and \(V_2\) the number of non-root vertices of degree one and two, respectively. Then:
\[ \abs{Stab_{\mathfrak{T}}(\mathcal{G})} = 2^{V_1+V_2} 
;. \]
\end{lem}
\begin{proof}
The orientation of the root vertex is held fixed. If a non-root vertex has degree one, twisting the edge incident to it does not change the ribbon graph---the twist is ``reducible''. If a non-root vertex has degree two, twisting both incident edges again does not change the ribbon graph. If both halfedges of a degree two vertex belong to the same edge, it is necessarily the root vertex, since \(\mathcal{G}\) is assumed to be connected. This is depicted in Fig.~\ref{fig: redtwist}.
\end{proof}
	
	\begin{figure}[ht]
		\centering
		\begin{tikzpicture}[scale=6/8]
			\node {\includegraphics[scale=6]{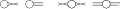}};
			\node at (-3.09,0) {\(\sim\)};
			\node at (2.39,0) {\(\sim\)};
		\end{tikzpicture}
		\caption{Reducible twists at vertices of degree one and two.}
		\label{fig: redtwist}
	\end{figure}
It follows that:
\begin{equation}\label{eq: doublesum}
\sum_{\substack{\mathcal{M} \\ \text{connected, rooted}}}
\frac{\abs{Stab_{\mathfrak{T}}(\smash[t]{\mathcal{M})}} }{ 2^{V}}
\sum_{[\mathcal{G}]\in Orb_{\mathfrak{T}}(\mathcal{M})}
\mathscr{A}(\mathcal{G})
= \sum_{\substack{\mathcal{M} \\ \text{connected, rooted}}}
\frac{1}{2^{V_{\geq 3}}} 
\sum_{[\mathcal{G}]\in Orb_{\mathfrak{T}}(\mathcal{M})}
\mathscr{A}(\mathcal{G})
\; ,
\end{equation}
where \(V_{\geq 3}\) denotes the number of non-root vertices of degree \(\geq 3\). In order to reshuffle this expression into a sum over rooted ribbon graphs, we recall that two ribbon graphs are equivalent if one can be obtained from the other by vertex re-embeddings. This implies that if two combinatorial maps \(\mathcal{M}_1\) and \(\mathcal{M}_2\) differ only by reversing the order of halfedges around some of their vertices then
$Orb_{\mathfrak{T}}(\mathcal{M}_1) = Orb_{\mathfrak{T}}(\mathcal{M}_2)$.
Reversing the order of halfedges at a vertex of degree lower than three is trivial hence for rooted ribbon graphs the multiplicity is \(2^{V_{\geq 3}}\).
As a result, the perturbative series of the two point function \(G_2\) (\ref{eq: DSE-G2}), for \(\kappa=0\) is:
\begin{equation*}
G_2(0,\lambda)=\sum_{[\mathcal{G}] \text{connected, rooted}}
\mathscr{A}(\mathcal{G}) 
; .
\end{equation*}
	
When taking the \(\varrho\)-edges into account, we recall that \(\mathfrak{T}\) acts trivial on them. Thus, it is sufficient to consider the ribbon graph \(\mathcal{G=\mathcal{G}\suprho - \mathcal{E}\suprho}\) obtained by deleting all the \(\varrho\)-edges of \(\mathcal{G}\suprho\).
However, when calculating \(\abs{Stab_{\mathfrak{T}}(\smash[t]{\mathcal{G}\suprho})}\) a subtlety arises: \(\mathcal{G=\mathcal{G}\suprho - \mathcal{E}\suprho}\) is not necessarily connected. The \(\varrho\)-edges can be used to induce a rooting at every connected component \(\mathcal{G}_c\subset \mathcal{G}\): Consider the connected components as effective vertices in a graph with only \(\varrho\)-edges; pick a spanning tree in that graph; from every \(\mathcal{G}_c\) there is a unique path in the tree to the original root; let the halfedge of \(\mathcal{G}_c\), belonging to that path, be another root. The stabilizer \(Stab_{\mathfrak{T}}(\mathcal{G})\) factors over the \(\mathcal{G}_c\) and using Lemma~\ref{thm: stab} for each rooted connected component one obtains:
\[
\abs{Stab_{\mathfrak{T}}(\mathcal{G})} = \prod_{\mathcal{G}_c\subset\mathcal{G}} 2^{V_1(\mathcal{G}_c)+V_2(\mathcal{G}_c)}
\; .
\]

One has to partially resum the double sum over combinatorial maps and ribbon graphs with \(\varrho\)-edges analogous to \eqref{eq: doublesum} into a sum over rooted ribbon graphs with \(\varrho\)-edges. The multiplicity of a ribbon graph with multiple rooted connected components is
$\prod_{\mathcal{G}_c\subset\mathcal{G}}2^{V_{\ge 3}(\mathcal{G}_c)}$
and one arrives at:
\[
\prod_{\mathcal{G}_c\subset\mathcal{G}}2^{V(\mathcal{G}_c)-1} = 2^{V(\mathcal{G})-C(\mathcal{G})}
\; ,
\]
where \(C\) denotes the number of connected components: the \(-1\) in the exponent appears because \(V_1\), \(V_2\) and \(V_{\geq 3}\) count only non-root vertices hence sum up to $ V(\mathcal{G}_c)-1$ in each connected component.

\begin{proof}[Proof of Proposition \ref{prom:combitwopoint}]
The discussion above goes through mutatis mutandis for multi ribbon graphs.
Combining \eqref{eq:partit} with \eqref{eq: DSE-tensor}, the perturbative series of the two point function writes:
\[
 G_2(\lambda) = \sum_{\mathcal{M} 
 \text{ connected, rooted}
 }
\frac{
\abs{Stab_{\mathfrak{T}}(\smash[t]{\mathcal{M})}} 
}{V(\mathcal{M})!\, 2^{V(\mathcal{M})-1}}  \Big(\prod_{v\in\mathcal{M}}  \frac{1}{
\prod_{q} \deg^q v!}\Big)  
\sum_{[\mathds{G}]\in Orb_{\mathfrak{T}}(\mathcal{M} )} \mathscr{A}(\mathds{G} ) \;,
\]
with $\mathcal{M}$ and edge multicolored combinatorial map. All objects in the above expression are fully labeled. Rooting prevents non-trivial symmetry factors, thus it is sufficient to count the ways to assign labels to a multi-ribbon graph:
1.~Pick a spanning tree;
2.~There are \(V!\) ways to label the vertices;
3.~At the root vertex \(v_0\), the root breaks the cyclicity of halfedges, thus there are \(\prod_{q}\deg^q v_0!\) ways to label the different types of multi-ribbon halfedges;
4.~At each non-root vertex one halfedge is part of the unique path in the tree towards the root; This again breaks cyclicity and there are \(\prod_{q}\deg^q v!\) ways to label the halfedges.
The amplitudes do not depend on the labeling, thus, in terms of unlabeled but rooted objects:
\[
G_2(\lambda) = \sum_{\substack{\mathcal{M} 
	\text{ connected,}\\ \text{rooted, unlabeled}}
}
\frac{
	\abs{Stab_{\mathfrak{T}}(\smash[t]{\mathcal{M})}} 
}{2^{V(\mathcal{M})-1}} 
\sum_{[\mathds{G}]\in Orb_{\mathfrak{T}}(\mathcal{M} )} \mathscr{A}(\mathds{G} ) 
=
\sum_{\lbrack\mathds{G}\rbrack\ \text{connected, rooted}}
\frac{1}{2^{C(\mathds{G}-\mathcal{E}\suprho )-1}} \mathscr{A}(\mathds{G}) \;, \]
where \(C(\mathds{G}-\mathcal{E}\suprho)\) counts the number of connected components of the multi-ribbon graph after deletion of the \(\varrho\)-edges.
\end{proof}

For example $G_2$ up to quadratic order in the coupling constants for $D=2$ is:
\[
	\begin{split}
		G_2(\kappa,\lambda)= N_1N_2 
		&-\lambda\Big(
		\tikzmarknode{rg1}{N_1 N_2} + \tikzmarknode{rg2}{N_1^2N_2} + \tikzmarknode{rg3}{N_1N_2^2}
		\Big)
		\\[19pt]
		&+\lambda^2\Big(
		\left(\tikzmarknode{rg6}{2}+\tikzmarknode{rg9}{2}+\tikzmarknode{rg8}{1}\right) N_1N_2 +
		\left(\tikzmarknode{rg5}{4}+\tikzmarknode{rg10}{1}\right)N_1^2N_2 +
		\left(\tikzmarknode{rg12}{4}+\tikzmarknode{rg14}{1}\right) N_1N_2^2 
		\\[19pt]
		&+ 
		\left(\tikzmarknode{rg11}{4}+\tikzmarknode{rg13}{1}\right) N_1^2N_2^2 +
		\tikzmarknode{rg7}{2 N_1N_2^3} +
		\tikzmarknode{rg4}{2 N_1^3N_2}
		\Big)
		\\[20pt]
		&-2\kappa \Big( \tikzmarknode{rg15}{N_1N_2} + \frac{1}{2}\tikzmarknode{rg16}{N_1^2N_2^2}\Big)
		\\[19pt]
		&+4\kappa^2\Big( (\tikzmarknode{rg17}{1}+\tikzmarknode{rg18}{2})N_1N_2 + \frac{1}{2}(\tikzmarknode{rg19}{4}+\tikzmarknode{rg20}{1}) N_1^2N_2^2 + \tikzmarknode{rg21}{\frac{2}{4}N_1^3N_2^3} \Big)
		\\[22pt]
		&+2\kappa\lambda\Big( (\tikzmarknode{rg22}{4}+\tikzmarknode{rg23}{2})N_1N_2 + (\tikzmarknode{rg24}{4}+\tikzmarknode{rg25}{2})N_1^2N_2
		\\[19pt]
		&+  (\tikzmarknode{rg26}{4}+\tikzmarknode{rg27}{2})N_1N_2^2 + \frac{4}{2}\tikzmarknode{rg28}{N_1^2N_2^2} + \frac{4}{2}\tikzmarknode{rg29}{N_1^3N_2^2} + \frac{4}{2}\tikzmarknode{rg30}{N_1^2N_2^3} \Big) + \dots\; .
	\end{split}
	\begin{tikzpicture}[overlay, remember picture, node distance=3pt and -8pt]
		\node [below=of rg1] {\includegraphics[scale=3]{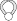}};
		\node [below=of rg2] {\includegraphics[scale=3]{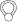}};
		\node [below=of rg3] {\includegraphics[scale=3]{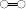}};
		\node [below=of rg4] {\includegraphics[scale=3]{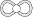}};
		\node [below=of rg5] {\includegraphics[scale=3]{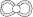}};
		\node [below left=of rg6] {\includegraphics[scale=3]{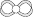}};
		\node [below=of rg7] {\includegraphics[scale=3]{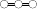}};
		\node [below right=of rg8] {\includegraphics[scale=3]{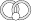}};
		\node [below=of rg9] {\includegraphics[scale=3]{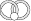}};
		\node [below right=of rg10] {\includegraphics[scale=3]{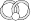}};
		\node [below=of rg11] {\includegraphics[scale=3]{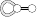}};
		\node [below=of rg12] {\includegraphics[scale=3]{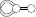}};
		\node [below right=of rg13] {\includegraphics[scale=3]{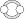}};
		\node [below right=of rg14] {\includegraphics[scale=3]{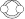}};
		\node [below=of rg15] {\includegraphics[scale=3]{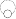}};
		\node [below=of rg16] {\includegraphics[scale=3]{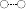}};
		\node [below left=of rg17] {\includegraphics[scale=3]{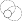}};
		\node [below=of rg18] {\includegraphics[scale=3]{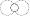}};
		\node [below left=of rg19] {\includegraphics[scale=3]{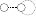}};
		\node [below=of rg20] {\includegraphics[scale=3]{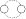}};
		\node [below=of rg21] {\includegraphics[scale=3]{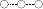}};
		\node [below left=of rg22] {\includegraphics[scale=3]{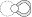}};
		\node [below=of rg23] {\includegraphics[scale=3]{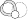}};
		\node [below left=of rg24] {\includegraphics[scale=3]{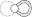}};
		\node [below=of rg25] {\includegraphics[scale=3]{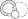}};
		\node [below left=of rg26] {\includegraphics[scale=3]{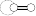}};
		\node [below=of rg27] {\includegraphics[scale=3]{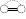}};
		\node [below=of rg28] {\includegraphics[scale=3]{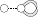}};
		\node [below=of rg29] {\includegraphics[scale=3]{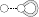}};
		\node [below=of rg30] {\includegraphics[scale=3]{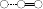}};
	\end{tikzpicture}
\]

\bigskip

Take for example the last three graphs. After deleting the \(\varrho\)-edges each splits into two connected components, this gives a factor \(\frac{1}{2}\) and in addition there are \(4\) distinct ways of rooting these graphs.
	
{\pagestyle{plain}
\printbibliography[heading=bibintoc,title=References]
}
\end{document}